\documentclass[3p]{elsarticle} 
 \myfooter[C]{}
\renewcommand{\thispagestyle}[1]{}

\makeatother

\date{}
\def\texpsfig#1#2#3{\vbox{\kern #3\hbox{\includegraphics{#1}\kern #2}}\typeout{(#1)}}
\usepackage{url,hyperref}
\hypersetup{
    colorlinks=true,
    linkcolor=blue,
    filecolor=magenta,      
    urlcolor=cyan,
    pdfpagemode=FullScreen,
    }
    
\usepackage{bbm}
\usepackage{eurosym}                           
\usepackage[latin1]{inputenc}                  
\usepackage{url}                               
\usepackage{longtable}                         
\usepackage{array}                             
\usepackage{graphicx,color}
\usepackage{amsthm}
\usepackage{amsbsy}
\usepackage{amssymb}
\usepackage{amsmath}
\usepackage{enumerate}
\usepackage{graphicx,color}
\usepackage{epsfig}
\usepackage[ruled,vlined]{algorithm2e}
\usepackage{multirow,bigdelim}
\usepackage[table]{xcolor}
\usepackage{natbib}
\usepackage{cleveref}
\usepackage{float}
\theoremstyle{plain}
\newtheorem{theorem}{Theorem}[section]

\newtheorem{dfn}[theorem]{Definition}

\newtheorem{example}{Example}[section]
\newtheorem*{rem}{Remark}
\theoremstyle{remark}

\theoremstyle{plain}
\newtheorem{lem}[theorem]{Lemma}

\theoremstyle{definition}

\newcommand{\e}{{\rm e}}        
\def\R{\mathbb{ R}}             
\def\E{\mathbb{ E}}             
\def\Q{\mathbb{ Q}}  
\def\N{\mathbb{ N}}  

\def\P{\mathbb{ P}}             

\def\F{\mathcal{F}}             

\def\Var{\mathbb{V}\text{ar}}   
\renewcommand{\d}{{\rm d}}      

\DeclareMathOperator*{\argmin}{arg\,min}

\def\dx{{\rm d}x}

\def\1{{\mathbbm{1}}}            

\theoremstyle{plain}

\usepackage[margin=1cm]{caption}
\numberwithin{equation}{section}	     
\title{\raggedright Volatility Parametrizations with Random Coefficients: Analytic Flexibility for Implied Volatility Surfaces }
\begin{document}
\author[1,3]{NICOLA F. ZAUGG \corref{cor1}}
\ead{N.F.Zaugg@uu.nl}
\author[1]{LEONARDO PEROTTI}
\ead{L.Perotti@uu.nl}
\author[1,2]{\raggedright LECH A. GRZELAK}
\ead{L.A.Grzelak@uu.nl}
\cortext[cor1]{Corresponding author.}
\address[1]{Mathematical Institute, Utrecht University, Utrecht, the Netherlands}
\address[2]{Financial Engineering, Rabobank, Utrecht, the Netherlands}
\address[3]{LGT Private Bank, Zurich, Switzerland}

\begin{abstract}
It is a market practice to express market-implied volatilities in some parametric form. The most popular parametrizations are based on or inspired by an underlying stochastic model, like the Heston model (SVI method [\textit{Quantitative Finance, 14 (2014), pp. 59-71}]) or the SABR model (SABR-parametrization [\textit{Wilmott Magazine, (2002), pp. 84-108}]). Their popularity is often driven by a closed-form representation enabling efficient calibration. However, these representations indirectly impose a model-specific volatility structure on observable market quotes. When the market's volatility does not follow the parametric model regime, the calibration procedure will fail or lead to extreme parameters, indicating inconsistency. This article addresses this critical limitation - we propose an arbitrage-free framework for letting the parameters from the parametric implied volatility formula be random. The method enhances the existing parametrizations and enables a significant widening of the spectrum of permissible shapes of implied volatilities while preserving analyticity and, therefore, computation efficiency. We demonstrate the effectiveness of the novel method on real data from short-term index and equity options, where the standard parametrizations fail to capture market dynamics. Our results show that the proposed method is particularly powerful in modeling the implied volatility curves of short expiry options preceding an earnings announcement, when the risk-neutral probability density function exhibits a bimodal form.
\noindent 
\end{abstract}

\maketitle
{\let\thefootnote\relax\footnotetext{The views expressed in this paper are the author's personal views and do not necessarily reflect the views or policies of their current or past employers. The authors have no competing interests.}}
{\let\thefootnote\relax\footnotetext{An accompanying Python code to the article containing examples is available on \href{https://github.com/NFZaugg/ImpliedVolatilityRandomization}{GitHub}.}}
\section{Introduction}
Obtaining a clean implied volatility surface from the options market is a fundamental aspect of modeling financial derivatives. Volatility surfaces reflect the current price level of vanilla option contracts and are often used as inputs to advanced derivative models. The models are then utilized to price exotic derivatives, set margin requirements for derivative trades, or are used by market makers to offer the most competitive price in the options market.
The essence of constructing an implied volatility surface is to encode the discrete and noisy option price quotes from the market into a clean, continuous surface, which provides a value of the Black-Scholes implied volatility for any desired combination of time to expiry and strike on a domain of interest. The main challenge is to obtain an implied volatility surface that is arbitrage-free, meaning that it does not lead to any static arbitrage opportunities for the implied option prices.

The existing literature on this problem is elaborate, and more innovative techniques are continuously being developed. This paper aims to provide a generic method to enable \emph{implied volatility surface parametrizations} to become more effective by means of \emph{parameter randomization}. It has been shown in the past \cite{jacquier2019randomized, jacquier2018black, github}, that randomization of model parameters in stochastic models can increase their flexibility and, therefore, enable a better fit to the market. We apply a randomization technique to improve the flexibility of implied volatility parametrizations and obtain a new parametrization with only limited additional parameters. The extra flexibility helps the parametrization to accurately model the market volatility when the classical parametrizations fail, such as before earnings announcements or generally for short-maturity options.
\subsection{Problem Setting}
We study the premium (or price) of a vanilla European option on an asset $S$ with a predetermined expiry date $T$ and strike price $K$. The price of such an option is written as
\begin{equation}
\label{eq:bs}
     V_{c/p}(T,K) = BS_{c/p}(t_0,S_0,T,K;\hat\sigma(T,K)),
\end{equation}
where $BS_{c/p}$ is the Black-Scholes formula for a call/put option and a constant risk-free rate $r$, $\hat\sigma(T,K) \geq 0 $ is called the \emph{(Black-Scholes) implied volatility} and $S_0$ the spot price of the asset. The Black-Scholes formula is a fundamental result in mathematical finance based on the Black-Scholes model, which assumes a lognormal probability distribution for the asset price process $S_t, t \geq t_0$. Since this assumption is known to be unrealistic, the Black-Scholes formula is most often used as a convenient way to quote prices of options available in the market as implied volatilities $\hat\sigma(T,K)$, rather than a pricing model. A market generally quotes options as implied volatilities on a discrete grid of strike prices and expiries. Suppose that $\Theta_{mkt}$ denotes a set of implied volatility quotes of the asset $S$ for $N$ expiries $\{T_1,T_2,\dots,T_{N}\}$ and $M$ strikes $\{K_1,K_2,\dots,K_M\}$:
\begin{equation}
    \Theta_{mkt} = \left\{\hat{\sigma}_{mkt}(T_n,K_m) : 1 \leq n \leq N , 1 \leq m \leq M\right\}.
\end{equation} 
Under the assumption of an efficient market, the market option quotes are set at a fair value based on market expectations of the future payoff. Practitioners are often interested in the price of an option for a pair $(T,K)$ whose price $\hat{\sigma}_{mkt}(T,K) \notin \Theta_{mkt}$ is not in the commonly traded set, for instance, as market makers attempt to set fair option prices. Since the market does not directly provide this value, the objective is to extend $\Theta_{mkt}$ to a smooth function $\hat\sigma(T,K) \colon \Pi \to \R_+$ of implied volatilities, which provides such a quote for any pair on a desired domain $\Pi \subset (t_0,\infty)\times \R_+$. Since $\hat\sigma(T,K)$ is a smooth function in two variables, it defines a smooth surface on $\R^3$ and is therefore called the \emph{implied volatility surface} (IV surface, or volatility surface). 
A primary condition of $\hat\sigma(T,K)$ is that the function is an extension of the discrete quotes, meaning that\begin{equation}
\label{eq:primary}
     \hat{\sigma}(T_n,K_m) = \hat{\sigma}_{mkt}(T_n,K_m), \quad 1 \leq n \leq N , 1 \leq m \leq M.
\end{equation}
Secondly, the resulting volatility surface should always be free of static arbitrage opportunities. With the Black-Scholes equation, the volatility surface defines the theoretical prices for put and call options on the entire domain. If these option prices present arbitrage opportunities by buying or selling options with different strikes and expiries, the volatility surface is not deemed suitable since market participants could exploit these, thereby eliminating the opportunities. The absence of arbitrage is guaranteed by several conditions that can be equivalently formulated on the option prices or directly on the volatility surface~\cite{guo2016generalized}.

Once the implied volatility function $\hat\sigma(T,K)$ is constructed from the market quotes and is arbitrage-free, it has various applications. Firstly, the option surface provides information on the market's expectation of the probability distribution of underlying asset $S$. For instance, the risk-neutral probability distribution of the asset $S$ can be derived with the Breeden-Litzenberger formula~\cite{breeden1978prices}, and the surface can be used to calibrate pricing models for exotic derivatives. Furthermore, the surface can be used for trading purposes by assessing volatility expectations, detecting arbitrage opportunities, or providing liquidity to the market.
\subsection{Literature Review}
Implied volatility surface modeling has been studied extensively in the past. Successful approaches have been developed using stochastic and local volatility models, statistical properties of the implied volatility surface, and interpolation schemes. As of late, new streams of data-driven approaches have been considered, further enhanced with machine learning tools~\cite{franccois2022venturing, ackerer2020deep, gonon2024operator,chataigner2020deep}. An extensive review of the classical methods is found in~\cite{homescu2011implied}. Below, we summarize the approaches that are relevant to our methodology.

The simplest way to obtain the function $\hat\sigma(T,K)$ is to choose an interpolation and extrapolation scheme to interpolate the available volatility quotes on the continuum. It is well known~\cite{fengler2009arbitrage} that such an approach is not arbitrage-free since a direct interpolation on the volatility surface often leads to a non-convex option pricing function. The mispricing is further exaggerated for assets with few option quotes available in the market, where the interpolation determines a large portion of the surface, increasing the risk of generating implausible or arbitrageable prices. An arbitrage-free interpolation scheme is feasible on the option pricing function rather than the implied volatilities~\cite{le2021arbitrage, andreasen2011volatility}. The interpolation methods remain problematic for assets for which only a few high-quality option prices are available. Furthermore, to obtain the implied volatilities from the market prices, an inversion of the Black-Scholes formula is required, which can only be achieved by a root-finding algorithm. 

The surfaces based on interpolation schemes have the feature that the market quotes $\Theta_{mkt}$ are always matched by $\hat\sigma(T,K)$ and therefore that (\ref{eq:primary}) is true. The market quotes on the surface are fixed first, and the interpolated points are added in addition to completing the surface. \emph{Volatility surface parametrizations} offer an alternative to this scheme. Rather than fixing the market quotes on the surface, the approach uses a fitting algorithm to choose the best-fitting surface from a predefined set of surfaces. First, a set of smooth surfaces $(T,K) \mapsto \hat\sigma(T,K; \overline p)$ are defined by parameters $\overline{p} = (p_1,p_2,\dots,p_m)$. The parametrization is chosen so that the surface is always free of arbitrage. The domain of the parameter space defines the set of all possible surfaces that can be created. The implied volatility surface for an asset $S$ is then chosen such that the difference to the market quotes $\Theta_{mkt}$ is minimized. While the surface will not exactly match $\Theta_{mkt}$, the difference is negligible if the space of possible surfaces is ``large" enough. 

Volatility parametrizations are often directly derived from a parametric asset model, such as a stochastic volatility model. Although a volatility surface can be generally derived for any asset price model, not all models offer a convenient analytical form for the volatility surface. Defining an analytically tractable and realistic model (in terms of fitting the market) is thus a challenging task. The \emph{Stochastic Volatility Inspired (SVI)} parametrization~\cite{gatheral2014arbitrage} was derived from a stochastic volatility model and enjoyed great success in the past. The parametrization was extended and improved in multiple works~\cite{mingone2022no, corbetta2019robust}. Another successful parametrization is the \emph{SABR parametrization}~\cite{Hagan:2002}, a direct result of the Hagan et al. formula provided by the SABR model (Stochastic-Alpha-Beta-Rho). Parametric volatility surfaces are convenient for numerous reasons. Firstly, since the parametrizations are derived from an asset model, the absence of arbitrage is guaranteed as long as the parameters remain in a well-defined range. Simple constraints on the parameter space are sufficient to ensure the absence of arbitrage. Secondly, the parameters can be attributed to the shape of the volatility surface. This allows for a simple way to compare surfaces and the expression of the dynamics of the surface in terms of the dynamics of the parameter. Thirdly, theoretically, no minimum number of market quotes is required to fit an arbitrage-free surface. Although a higher number of quotes is desired for stability in the calibration, a parametrization can be fit to any number of quotes and remain arbitrage-free. This is of particular importance for assets with only a few quotes or only in a particular strike region, where little to no information is provided for the tails of the volatility curves. 

The main drawback of parametrizations is that there is no guarantee that a surface that is appropriately close to the market data can be found. The flexibility of the surface is limited by the chosen parametrization, which does not cover all viable (i.e. arbitrage-free) volatility shapes. If the market conditions are not within the scope of the parametrization, the fitting process fails, and the parametrization misrepresents the actual market conditions. These out-of-scope market conditions can be systemic, meaning that parametric models are generally unfit to model a certain market behavior. For instance, it is well-known that short-maturity options in equity markets exhibit a steeper at-the-money implied volatility term structure than attainable under regular stochastic volatility models~\cite{gatheral2011volatility}. Alternatively, the market condition can be out of scope due to irregular behavior in the market, such as spikes in volatility or higher-than-usual uncertainty. It is not uncommon for very short-term options markets (near expiry) to exhibit a W-shaped volatility shape, or ``mustache" shapes~\cite{glasserman2023w} before an earnings announcement of the underlying equity. These shapes of volatility surface arise from \emph{bimodal} risk-neutral probability density functions for the stock price, which reflect the dichotomous nature of the earnings, which can either have a positive or negative effect on the asset price. Classical diffusion models\footnote{It needs to be mentioned that SVI-type parametrizations do not rely on the an underlying model, but their development was merely inspired by volatility shapes of the Heston model. Nevertheless, one observes the same characteristics in terms of fit as classical stochastic volatility models.} are unable to create such shapes due to the single continuous diffusive driver, and the corresponding volatility parametrization is limited in such cases. This limitation highlights the industry's need for more flexible parameterized surfaces to achieve better fits in such market situations.

\subsection{Contributions}
In this paper, we study volatility surfaces in cases when the parametric surfaces reach the limit of their flexibility. We introduce a generic method to enhance the flexibility in such cases using a randomization scheme of the parameters of the surface. By replacing one or more parameters of the parametric surface with a random variable, we derive a new implied volatility surface with increased flexibility to fit the market. The additional flexibility stems from an induced mixing probability density, which is created by the stochastic nature of the randomized parameter. Mixing techniques have been studied in the past and have been shown to fit well to the market~\cite{rebonato2004unconstrained, bloch2011smiling, wilkens2005option}. The new randomized volatility surface is expressed in terms of the original parameters plus parameters to specify the chosen random variable. To show the effectiveness of the novel method, we apply the method to the existing parametrizations of the SABR parametrization and demonstrate that the additional flexibility enables them to better fit the market data for SPX options with expiry dates of less than half a year. Furthermore, we derive a second kind of randomization where we randomize the spot parameter of the asset. This arbitrage-free randomization is shown to be particularly effective for near-maturity options (i.e., options with an expiry date in the next few days) before an earning announcement, which induces a multi-modal-type volatility regime.\footnote{We speak of a multi-modal-type volatility regime when the risk-neutral probability density functions implied by the volatility surface exhibit more than one mode and are therefore \emph{multi-modal}.} The paper is organized in the following way:
\Cref{sec:model} describes the process of randomizing a parametrization of volatility surfaces and proves that the randomization is indeed free of static arbitrage. We present a few examples of randomized parametrizations and their effectiveness in \Cref{sec:numerics}. In particular, we show an improved fit to options on the S\&P 500 index. In \Cref{sec:rand-Spot}, we then introduce the randomized spot parametrization and also present an example of its use on short-maturity AMZM (Amazon.com Inc) options on the day of an earnings announcement. Finally, we conclude the paper in \Cref{sec:conclusion}.

\section{Parameter Randomization and Analytical Implied Volatility}
\label{sec:model}
\subsection{General Volatility Parametrizations}
We consider a type of implied volatility surfaces, which are given in parametric form, defined by a set of parameters. As a starting point, suppose that $\overline{p} := (p_1,\dots,p_m)$ is an array of $m$ constant parameters on a parameter domain $\overline{p} \in \mathcal{D}\subset \R^m$. We define a \emph{parameterized implied volatility function} as a positive $\mathcal{C}^{1,2}$ function
\begin{equation}
    \hat\sigma(T,K;\overline{p}):=\hat\sigma(T,K;(p_1,\dots,p_m)) \geq 0,
\end{equation} of expiry and strike price $(T,K) \in \Pi = (t_0,\infty) \times [0,\infty)$ given the associated parameters $p_1,\dots,p_m$. The map $\hat\sigma(T,K;\overline{p})$ for a fixed $\overline{p}$ is referred to as a \emph{parameterized implied volatility surface}\footnote{Sometimes we refer to the function as \emph{parametric volatility surface} or simply \emph{parametrization}.} since the set \begin{equation}
    \left\{ \left( T,K,\hat\sigma(T,K;\overline{p})\right): (T,K) \in \Pi\right\} \subset \R^3,
\end{equation} defines a surface over the domain $\Pi$. The exact shape of the surface depends on the parameters $\overline{p}$ and the chosen parametrization function. 
Since the parameterized implied volatility function is positive, we can compose the function with the Black-Scholes formula and obtain the \emph{parametrized (put/call) pricing function}, defined as
 \[ V_{c/p}(T,K; \overline{p}) := BS_{c/p}(t_0,S_0,T,K;\hat\sigma(T,K; \overline{p})),\]
 for an asset with spot price $S_0$ at $t_0$. The function provides a continuum of European call and put option prices on the underlying asset on $(t_0 \times \infty) \times [0,\infty)$ given the implied volatility $\hat\sigma(T,K; \overline{p})$. We can extend the functions to the limit points at $T=t_0$, which, under certain conditions exist and are given by $V_c(t_0,K; \overline{p}) = (S_0-K)^+$ and $V_c(t_0,K; \overline{p}) = (K-S_0)^+$ (see \Cref{def:abitrage_free}). The pricing function defines two additional surfaces in $\R^3$ over the domain $\Pi$, which we refer to as the (put/call) pricing surface. 
 
 Since the parameterized implied volatility and pricing surface define a set of prices for tradable options, these prices must be free of arbitrage opportunities. There are various ways to define the absence of arbitrage of an implied volatility surface, most of which are equivalent. Here, we will utilize the \emph{model-free} definition~\cite{roper2009implied}, which does not rely on the introduction of stochastic asset models but can be expressed as a set of conditions on the pricing surfaces. 
\begin{dfn}[Arbitrage-free volatility surface]
\label{def:abitrage_free}
    Given a set of parameters $\overline{p}$, let $\hat\sigma(T,K; \overline{p})$ be a parametrized implied volatility surface defined on $\Pi = \Pi_T \times \Pi_K = (t_0,\infty)\times \R_+$. Let $V_c(T,K; \overline{p})$ be its call pricing function, which is extended to the limit points at $T=t_0$. The parametrization is called free of ``butterfly" arbitrage if the following conditions hold on the call pricing function and a constant $s > 0$:
\begin{enumerate}[i)]
    \itemsep0em
    \item $V_c(T, \cdot ;\overline{p} )$ is convex and non-increasing for all $T \in \Pi_T$.
     \item $\lim_{K \to \infty } V_c(T,K;\overline{p}) = 0 $ for all $T \in \Pi_T$.
     \item $(s - K)^+ \leq V_c(T,K;\overline{p}) \leq s$ for all $(T,K) \in \Pi$.
     \item $V_c(t_0,K;\overline{p}) =   (s - K)^+$ for all $K \in \Pi_K$.
\end{enumerate}
If the following additional condition holds, the pricing surface is also free of ``calendar" arbitrage, and we call the surface arbitrage-free:
\begin{enumerate}[i)]
\setcounter{enumi}{4}
\itemsep0em
         \item $V_c(\cdot,K;\overline{p})$ is non-decreasing for all $K \in \Pi_K$. 
\end{enumerate}
\end{dfn}
Equivalent conditions in terms of the put-call parity can be derived for the put pricing function $V_p(T,K;\overline{p})$, or directly in terms of the volatility surface $\hat \sigma (T,K;\overline{p})$~\cite{guo2016generalized}. Under these conditions, one can prove the existence of a non-negative local martingale process on a suitable probability space such that the call price function can be written as a risk-neutral expectation of the final payoff. In other words, an arbitrage-free model (price process) exists that yields the volatility surface. However, we will not utilize this fact to avoid introducing any model dynamics and remain ``model-free". 

Due to the continuity of the Black-Scholes formula with respect to the volatility input, and since we chose the volatility surface to be $\mathcal{C}^{1,2}$, it is necessary for the parameterized price function $V_{c/p}(T,K ; \overline{p})$ to be $\mathcal{C}^{1,2}$. Therefore, one can obtain the risk-neutral probability density functions $\{p_{S_t} \colon  t \geq t_0\}$, defined as\footnote{Note that one can also use the put prices $V_{p}(t,x; \overline{p})$ to obtain the same result.}
\begin{equation}
    \label{eq:breeden-liz}
    p_{S_t}(x) := \e^{r(t-t_0)} \frac{\d^2 V_{c}(t,x; \overline{p})}{\d x^2},
\end{equation} 
for any $t \geq t_0$ using the Breeden-Litzenberger formula~\cite{breeden1978prices} and risk-free rate $r$. In particular, one can show with a few manipulations that the expected value of the random variable $S_t$ whose PDF is given by is $p_{S_t}(\cdot)$, is
\begin{equation}
\label{eq:expectation_forward}
    \E[S_t] = S_0\e^{r(t-t_0)}.
\end{equation} Conversely, given a set of PDFs, such that (\ref{eq:expectation_forward}) holds for all $t \geq t_0$, we can define an arbitrage-free pricing surface
by integration of the probability densities~\cite{gatheral2014arbitrage}.

    If considering an implied volatility surface at a fixed expiry $T \in \Pi_T$, we refer to the resulting function $\hat\sigma_{T}(K; \overline{p}) := \hat\sigma(T,K;(p_1,\dots,p_m))$ as an \emph{implied volatility slice}.
Furthermore, we denote the collection of all possible volatility surfaces $\mathcal{S}$ for a parametrization as 
\begin{equation}
    \mathcal{S} := \left\{ \hat\sigma(T,K;\overline{p}) : \overline{p} \in \mathcal{D} \right\}.
\end{equation} 
To obtain a volatility surface that fits the market, we find the optimal parameters $\overline{p}$, which minimizes the difference between the market quotes
\begin{equation}
\label{eq:optimization}
    \overline{p}_{opt}  = \argmin_{ \hat\sigma(T,K;\overline{p}) \in \mathcal{S} } \sum_{i,j}^{N,M} \lVert{\hat\sigma(T_i,K_j;\overline{p}) - \sigma_{mkt}(T_i,K_j)}\rVert,
\end{equation}
with respect to a desired norm $\lVert.\rVert$. The ability to fit the market, therefore, directly depends on the ``size" of $\mathcal{S}$, whether $\mathcal{S}$ contains a function that can match the market.
\subsection{Volatility Parametrizations with Random Parameters}
 In an effort to increase the size of $\mathcal{S}$ and obtain better fitting arbitrage-free volatility surfaces, we consider the possibility of adding stochasticity to the implied volatility surface by replacing one of the parameters $p_i$ for $1\leq i\leq m$ with a random variable. Let $\hat\sigma(T,K;\overline{p})$ be a parameterized implied volatility function with parameter domain $\mathcal{D} \subset \R^m$ and suppose that $\vartheta$ is a real-valued random variable on a probability space $(\Omega,\F, \P)$. We can replace the $i$-th entry of $\overline{p}$ with the random variable $\vartheta$ to obtain the random vector $\overline{p}(\vartheta) = (p_1,p_2,\dots,\vartheta,\dots, p_m)$. If we assume that $\overline{p}(\vartheta)$ is almost surely in the domain $\mathcal{D}$, a realization $\omega \in \Omega$ determines a real-valued vector $\overline{p}(\theta):= (p_1,p_2,\dots,\theta,\dots, p_m) \in \R^m$ with $\theta = \vartheta(\omega)$, which almost surely provides an arbitrage-free implied volatility surface $\hat\sigma(T,K;\overline{p}(\theta))$ and its associated pricing functions $V_{c/p}(T,K; \overline{p}(\theta))$. 
 
Assuming measurability of the parametrization function in the parameters, the pricing function $V_{c/p}(T,K; \overline{p}(\vartheta))$ is thus a random variable, and we can compute the expectation as the Lebesgue-integral
\begin{equation}
\label{eq:randPrice}
\E[V_{c/p}(T,K; \overline{p}(\vartheta))]=\int BS_{c/p}(t_0,S_0,T,K;\hat\sigma(T,K;\overline{p}(\vartheta))) \d \P.
\end{equation} 
If $\vartheta$ is absolutely continuous, the integral is a Riemann integral
\begin{equation}
\E[V_{c/p}(T,K; \overline{p}(\vartheta))]=\int_{\R} BS_{c/p}(t_0,S_0,T,K;\hat\sigma(T,K;\overline{p}(\theta)) f_{\vartheta}(\theta)\d \theta, 
\end{equation}
where $f_{\vartheta}$ is the associated probability density function of $\vartheta$. The equation shows that the expected option price under the randomized volatility surface is effectively an average of option prices on the sample space of the parameters of the volatility surface. Since the expectation of the randomized pricing function is a deterministic function in two variables, we can investigate its suitability as a pricing function. Due to the linear properties of the integral operator, the function is indeed an arbitrage-free pricing function.
\begin{lem}[Arbitrage-free randomization]
\label{lem:arb-free}
Let $\hat\sigma(T,K; \overline{p})$ be an arbitrage-free parameterized implied volatility function on the parameter domain $\overline{p} \in \mathcal{D}$ and let $\vartheta$ be a real-valued random variable on a probability space $(\Omega,\F, \P)$ and denote $\mu(\cdot)$ as its law.

Suppose that $\overline{p}(\vartheta) = (p_1,p_2,\dots,\vartheta,\dots,p_m)$ is the random vector where we replaced the $i$-th parameter of $\overline{p}$ with $\vartheta$. Then, the map $\hat\sigma(T,K)$ such that
    \begin{equation}
    \label{rand:weak}
       BS_{c/p}(t_0,S_0,T,K;\hat\sigma(T,K))  = \E[V_{c/p}(T,K; \overline{p}(\vartheta))], \;\; \forall (T,K) \in \Pi,
    \end{equation}
    is an arbitrage-free implied volatility surface. We refer to $\hat\sigma(T,K)$ as a randomization of $\hat\sigma(T,K; \overline{p})$ with random variable $\vartheta$ in parameter $i$.
\end{lem}
\begin{proof}
The conditions of an arbitrage-free surface are given by \Cref{def:abitrage_free}, and it is possible to prove the lemma by confirming that each condition is closed under convex combinations. We will show separately that the surface is free of butterfly arbitrage and free of calendar arbitrage. The first step is thus to confirm that for each fixed $T \in \Pi_T$, the volatility slice $\hat{\sigma}_T(K)$ is arbitrage-free. 
We recall that the call price can be written as an integration of the risk-neutral PDF $p_{S_T;\theta}$ 
    \[V_{c}(T,K; \overline{p}(\theta)) = \e^{-r (T-t_0)}\int_{K}^\infty (x - K) p_{S_T;\theta}(x)\d x.\]
    The expectation can thus be written as
    \begin{align*}
        \E[V_{c}(T,K; \overline{p}(\vartheta))] &= \int_{\R} V_{c}(T,K; \overline{p}(\theta)) \d \mu(\theta)\\
        &= \int_{\R} \e^{-r (T-t_0)}\int_{K}^\infty (x - K) p_{S_T; \theta}(x) \d x \d \mu(\theta) \\
        &= \e^{-r (T-t_0)}\int_{K}^\infty (x - K) \int_{\R} p_{S_T; \theta}(x) \d \mu(\theta) \d x . 
    \end{align*}
    Since the random vector $\overline{p}(\vartheta)$ takes values almost surely in $\mathcal{D}$, the function $p_{S_T; \vartheta}(\cdot)$ is almost surely a probability density function. The expression $\int_{\R} p_{S_T; \theta}(x) \d \mu(\theta)$ is a convex combination of probability densities (mixture density), which we define as the function 
    \[\overline{f}(x):= \int_{\R} p_{S_T; \theta}(x)\d \mu(\theta).\]
    We claim $\overline{f}(x)$ to be a proper probability density function on $[0,\infty)$. The function $\overline{f}(x)$ is certainly positive since $p_{S_T; \theta}(\cdot)$ is a positive function, and integrates to 1 since
    \begin{align*}
 \int_0^\infty \overline{f}(x) \d x &= \int_0^\infty\int_{\R} p_{S_T; \theta}(x) \d \mu(\theta)\d x \\
 &= \int_{\R} \int_0^\infty\ p_{S_T; \theta}(x) \d x \d \mu(\theta) \\
 &=  \int_{\R} \d \mu(\theta)\\
 &=1.
    \end{align*}
    This show the sufficient condition for $\hat\sigma_T(K)$ to be butterfly-arbitrage-free \cite{gatheral2014arbitrage}.
    It remains to show that condition v) is true as well. Let $K$ be fixed and suppose that $t_1 < t_2 \in \Pi_T$. Since $\hat\sigma(T,K; \overline{p}(\theta))$ is arbitrage-free, we have that $V_c(t_1,K;\overline{p}(\theta)) \leq V_c(t_2,K; \overline{p}(\theta))$ for $\overline{p}(\theta) \in \mathcal{D}$. Since this is almost surely the case, we conclude that
    \[\E[V_c(t_1,K;\overline{p}(\vartheta))] \leq \E[V_c(t_2,K;\overline{p}(\vartheta))],\]
    and therefore, that condition v) is fulfilled, too.
\end{proof}
The map $\hat\sigma(T,K)$ as defined \Cref{rand:weak} is thus an arbitrage-free implied volatility surface, and the pricing of European-style options will collapse to determining a mixture distribution weighted with the probability distribution of $\vartheta$. If we specify the probability distribution of the random variable $\vartheta$ in parametric form, we can combine the sets of the parametrized surface and the random variable into a common set of parameters. Suppose that the random variable $\vartheta$ is given in parametric form by the parameters $\overline{q} = (q_1,q_2,\dots,q_l) \in \R^l$. We combine the parameter sets of $\overline{q}$ and $\overline{p}$ to an extended parameter parameter vector \begin{equation}
    \overline p^* = (p_1,p_2,\dots,p_{i-1},p_{i+1}, \dots,p_m, q_1, \dots,q_l) \in \mathcal{D}^*,
\end{equation} 
where $\mathcal{D}^*$ is a new parameter space. The map $\hat\sigma(T,K) := \hat\sigma(T,K;\overline p^* )$ therefore regains a parameterized form. 
\begin{example}
    We clarify the notation of the extended parameter space $\mathcal{D}^*$ on a simple example of a randomization. Suppose that $\hat\sigma(T,K,(a,b,c))$ is a parametric implied volatility surface with domain $(a,b,c) \in \mathcal{D} = \R^3$. We consider a randomization of $\hat\sigma(T,K,(a,b,c))$ in parameter $c$ with the normally distributed random variable $\vartheta \sim \mathcal{N}(\mu,\nu)$ for $(q_1,q_2) = (\mu,\nu) \in \R \times \R_+$. This randomization is thus given in parametric form by the function $\hat\sigma(T,K,\overline{p}^*)$ for the parameters 
    \begin{equation*}
        \overline{p}^* = (a,b,\mu,\nu) \in \mathcal{D}^* = \R^3 \times \R_+.
    \end{equation*}
The definition of $\hat\sigma(T,K,\overline{p}^*)$ is given by (\ref{rand:weak}).
\end{example}
Since we defined $\hat\sigma(T,K,\overline{p}^*)$ in terms of the pricing surface and the Black-Scholes formula has no analytical inverse in the volatility parameter, we no longer have an analytical expression of the volatility surface. However, it is possible to derive an analytic expansion of the implied volatility surface, which we aim to derive step by step for the remainder of the section.

As the pricing surface of the randomized parametrization $V_{c/p}(T,K;\overline{p}^*)$ is defined as an integral over the domain of parameters, weighted by the law of the random variable $\vartheta$, we now derive a discretization of the Lebesgue integral using a numerical integration scheme. While the discretization acts an approximation of the distribution of $\vartheta$ and the pricing surface, we will ensure that the discretization remains arbitrage-free \emph{for any degree of accuracy of the approximation}. This is crucial as the approximation will be the primary pricing surface parametrization to fit to the market, and the degree of approximation will act as an additional parameter to the surface. To ensure this property we choose a numerical integration scheme that maintains the convex property of the summation. The Gaussian quadrature is a suitable tool for such an approximation. 

Suppose we wish to integrate an expression $\int_a^b h(x) w(x) \d x$, on a real domain $[a,b] \subset \R$, for an integrable function $h(x)$ and a weight function $w(x)$ (a positive function which integrates to 1). The Gaussian quadrature approximates this as
\begin{equation}
    \int_a^b h(x) w(x) \d x \approx \sum_{n=1}^{N_q} \lambda_n h(x_n),
\end{equation}
 where $\{\lambda_n,\theta_n\}_{n=1}^{N_q}$ are known as the Gauss-quadrature weights. The important property of the Gauss-quadrature integration is that the integral is exact for any polynomial of degree less than $2N_q$. Since the expression $h(x)=1$ is a polynomial of degree $0$, the approximation is exact for this integral for any $N_q$\footnote{Other integration methods do not have this property. For instance, a trapezoidal method with only two grid points will yield $0$ if $w(x)$ has finite support.}, and we obtain
 \begin{equation}
 \label{eq:approx_weight_function}
    \int_a^b 1 w(x) \d x  = \int_a^b w(x) \d x =  \sum_{n=1}^{N_q} \lambda_n =1.
\end{equation}
In particular, this means that for an absolutely continuous real-valued random variable $X$, where the probability density function $w(x) = f_X(x)$ is a weight function, we can approximate its expectation as 
\begin{equation}
\E[g(X)] = \int_{\R} g(x) f_X(x) \dx \approx\sum_{n=1}^{N_q}\lambda_n g(x_n)
\end{equation}
where the pairs $\{\lambda_n,x_n\}_{n=1}^{N_q}$ are the Gauss-quadrature weights and nodes. The Gauss quadrature works more generally, as outlined in \ref{sec:appendix_qp}, to approximate the expectation $\E[g(X)]$ for any real-valued random variable. Applied to the randomized surface, we find
\begin{equation}
\label{eqn:genericPricing}
V_{c/p}(T,K;\overline{p}^*) =  \int_{\R} V_{c/p}(T,K;\overline{p}(\theta)) \d \mu(\theta) \approx\sum_{n=1}^{N_q}\lambda_n V_{c/p}(T,K;\overline{p}(\theta_n)),
\end{equation}
where the pairs $\{\lambda_n,\theta_n\}_{n=1}^{N_q}$ are the Gauss-quadrature weights and nodes depend on the distribution function $F_{\vartheta}(\cdot)$ of $\vartheta$ (see \ref{sec:appendix_qp}). We define the \emph{truncated}, or \emph{discretized pricing surface} with $N_q$ terms as the sum
\begin{equation}
    \label{eq:model_with_extra_p}
    V_{c/p}(T,K; \overline p^*, N_q) :=  \sum_{n=1}^{N_q}\lambda_n V_{c/p}(T,K;\overline{p}(\theta_n)).
\end{equation}
Since the discretization approximates $V_{c/p}(T,K; \overline p^*, N_q)$, and due to the convex property, it is arbitrage-free for any $N_q$.
\begin{lem}
    Let $V_{c/p}(T,K;\overline{p}^*)$ be a randomized pricing surface for a parametrization $\hat\sigma(T,K;\overline{p})$ with random variable $\vartheta$ in a parameter $i \leq m$. Let $V_{c/p}(T,K; \overline p^*, N_q)$ be the discretized pricing surface of the randomization. Then, there is a discrete random variable $\bar{\vartheta}$, which defines a randomization of $\hat\sigma(T,K;\overline{p})$ in parameter $i$ with the property that
    \begin{equation}
    \label{eq:equivalence_disc_cont}
    \E[V_{c/p}(T,K; \overline{p}(\bar\vartheta))] \equiv V_{c/p}(T,K; \overline p^*, N_q).\end{equation}
    This implies in particular that $V_{c/p}(T,K; \overline p^*, N_q)$ is an arbitrage-free pricing surface.
\end{lem}
\begin{proof}
    Let $\{\lambda_n,\theta_n\}_{n=1}^{N_q}$ be the Gauss-quadrature pairs of $\vartheta$ of degree $N_q$ and define the function
    \[ f_{\bar{\vartheta}}(x) = \begin{cases}
        \lambda_n, \quad &\text{ if } x = \theta_n,\\
        0 \quad &\text{ otherwise.}
    \end{cases}\]
    We claim that this is a probability mass function. Since $\lambda_n \geq 0$ for all $n \leq N_q$, the function is certainly positive. Furthermore, since quadrature integration is exact for polynomials up to $2N_q$ degree, we have as in (\ref{eq:approx_weight_function}), that $\sum_{n=1}^{N_q} \lambda_n =1$.
    By Skorohod's representation theorem, there exists a random variable $\bar\vartheta$, such that $f_{\bar{\vartheta}}$ is its probability mass function. 
    It follows that (\ref{eq:equivalence_disc_cont}) is true, since $\E[V_{c/p}(T,K;\overline{p}(\bar\vartheta))]$ is given by $\sum_{n=1}^{N_q}\lambda_n V_{c/p}(T,K; \overline{p}(\theta_n))$. 
\end{proof}
\subsection{Analytic Expansion of Randomized Volatility Surface}
 With the additional parameter and the help of the quadrature integration, we transformed the semi-analytical pricing surface into a finite sum of prices given by the discretized randomization of (\ref{eq:model_with_extra_p}). Since the Black-Scholes equation has no analytical inverse in the volatility argument, the implied volatility surface $\hat\sigma(T,K;\overline{p}^*,N_q)$ cannot be obtained explicitly. Nevertheless, Brigo and Mercurio ~\cite{brigo2002lognormal} showed how to get a polynomial expansion of the implied volatility function based on the Taylor expansion under a mixture model of lognormal prices. We derive a generalization of this result, and we use it to obtain an analytic expression of $\hat\sigma(T,K;\overline{p}^*,N_q)$. Let $(T,K)\in \Pi$ be fixed and suppose we want to find the value of $\hat\sigma(T,K;\overline{p}^*,N_q)$. We first define the function $m(T,K)$ as
\begin{equation}
    m(T,K) := \log \frac{S_0}{K} + rT,
\end{equation}
where $r$ is the interest rate. Furthermore, let $P(m)\colon \R \to \R_+$ a positive, continuously differentiable function of one variable $m$, for which we define the value at $m= m(T,K)$ as $P(m(T,K)) = \hat\sigma(T,K;\overline{p}^*,N_q)$. The Black-Scholes price of an option at $(T,K)$ with implied volatility $P(m(T,K))$ can be written as a function $f$, such that
\begin{align}
\label{eq:left-side}
    f(m,P(m)):= S_0 \bigg[\Phi\bigg(\frac{m+\frac12 P^2(m) T}{P(m)\sqrt{T}}\bigg)-\e^{-m}\Phi\bigg(\frac{m-\frac12P^2(m) T}{P(m)\sqrt{T}}\bigg)\bigg].
\end{align}
On the other hand, a discretized price surface $V_{c/p}(T,K;\overline{p}^*,N_q)$ with $\eta_n = \hat\sigma(T,K,\overline{p}(\theta_n))$ can be written as a function of $m$ as well. We define:
\begin{equation}
\label{eq:right-side}
    g(m) := S_0 \sum_{n=1}^{N_q}\lambda_n \bigg[\Phi\bigg(\frac{m+\frac12\eta_n^2 T}{\eta_n\sqrt{T}}\bigg)-\e^{-m}\Phi\bigg(\frac{m-\frac12\eta_n^2 T}{\eta_n\sqrt{T}}\bigg)\bigg].
\end{equation}
 If we equate the two equations, the equation
\begin{equation}
\label{eq: SpecialEq0}
    f(m,P(m)) = g(m),
\end{equation}
is exactly equal to (\ref{eq:model_with_extra_p}) at $m = m(T,K)$. The goal is to obtain a polynomial expansion of the function $P(m)$, which we can then evaluate at $m = m(T,K)$ to obtain the randomized volatility function at $(T,K)$. The expansion utilizes the ideas from the implicit function theorem~\cite{krantz2002implicit} to obtain any higher-order derivative from an implicit function, such as \Cref{eq: SpecialEq0}. We can obtain the derivatives of $P(m)$ by differentiating both sides with respect to $m$ and write the function $P(m)$ as a Taylor expansion in terms of the derivatives at $0$ and the function $P(m)$ at $m=0$. Prior to stating the full expansion, we show a property of the implicit equation, which simplifies the expansion terms. 
\begin{lem}
\label{lem:der_vanish}
    Let the functions $f$ and $g$ be defined as \Cref{eq:left-side} and \Cref{eq:right-side}. The set $S = \{(x,y) \in \R^2 : f(x,y) = g(x)\}$ is symmetric over the $y$-axis.
\end{lem}
\begin{proof}
    We will show directly that if $(x,y) \in S$, then the point $(-x,y)$ is also contained in $S$. Assuming $(x,y) \in S$, we have
    \[\Phi\bigg(\frac{x+\frac12 y T}{y\sqrt{T}}\bigg)-\e^{-x}\Phi\bigg(\frac{x-\frac12y T}{y\sqrt{T}}\bigg)= \sum_{n=1}^{N_q}\lambda_n \bigg[\Phi\bigg(\frac{x+\frac12\eta_n^2 T}{\eta_n\sqrt{T}}\bigg)-\e^{-x}\Phi\bigg(\frac{x-\frac12\eta_n^2 T}{\eta_n\sqrt{T}}\bigg)\bigg] .\]
    Multiplying both sides with $e^{x}$ and adding a term $1 - e^{x}$ on both sides we find
    \[1 - \e^{x} + \e^{x}\Phi\bigg(\frac{x+\frac12 y T}{y\sqrt{T}}\bigg)-\Phi\bigg(\frac{x-\frac12y T}{y\sqrt{T}}\bigg)= 1 - \e^{x}  + \sum_{n=1}^{N_q}\lambda_n \bigg[ \e^{x} \Phi\bigg(\frac{x+\frac12\eta_n^2 T}{\eta_n\sqrt{T}}\bigg)-\Phi\bigg(\frac{x-\frac12\eta_n^2 T}{\eta_n\sqrt{T}}\bigg)\bigg] .\]
    On the other hand, for $(-x,y)$, we have
    \begin{align*}
        \frac{1}{S_0}f(-x,y) &= \Phi\bigg(\frac{-x+\frac12 y T}{y\sqrt{T}}\bigg)-\e^{x}\Phi\bigg(\frac{-x-\frac12y T}{y\sqrt{T}}\bigg) \\
        &= 1- e^{x}  - \Phi\bigg(\frac{x-\frac12 y T}{y\sqrt{T}}\bigg) +\e^{x} \Phi\bigg(\frac{x+\frac12y T}{y\sqrt{T}}\bigg) 
    \end{align*}
    as well as
    \begin{align*}
        \frac{1}{S_0}g(-x) &=
    \sum_{n=1}^{N_q}\lambda_n \bigg[\Phi\bigg(\frac{-x+\frac12\eta_n^2 T}{\eta_n\sqrt{T}}\bigg)-\e^{x}\Phi\bigg(\frac{-x-\frac12\eta_n^2 T}{\eta_n\sqrt{T}}\bigg)\bigg]  \\
    &= 1 - e^{x} + \sum_{n=1}^{N_q}\lambda_n \bigg[- \Phi\bigg(\frac{x-\frac12\eta_n^2 T}{\eta_n\sqrt{T}}\bigg) +\e^{x} \Phi\bigg(\frac{x+\frac12\eta_n^2 T}{\eta_n\sqrt{T}}\bigg)\bigg]
    \end{align*}
    Equating the two terms we obtain that
    \[\frac{1}{S_0}f(-x,y) = \frac{1}{S_0}g(-x)   \iff f(-x,y) = g(-x),\]
    and we therefore conclude that $(-x,y) \in S$.
\end{proof}
From the lemma we make the useful conclusion that implicit function $P(m)$ is even, and has therefore vanishing odd expansion coefficients.
\begin{theorem}
\label{thm:expansion}
Let $\hat\sigma(T,K; \overline{p})$ be an implied volatility parametrization on the parameter domain $\mathcal{D}$, which is continuously differentiable in $T$ and $K$, and let $\hat\sigma(T,K; \overline p^*, N_q)$ with $\overline p^* \in \mathcal{D}^*$ be its discretized randomization with the quadrature points $\{\lambda_n, \theta_n\}_{n=1}^{N_q}$. Let \[m(T,K) := \log \frac{S_0}{K} + rT,\] be the log-moneyness with interest-rate $r$. 
Then, the discretized randomization $\hat\sigma(T,K; \overline{p}^*, N_q)$ at $(T,K)$ is given by the Taylor expansion $P_{(T,K)}(m)\colon A \subset \R \to \R$, evaluated at $m= m(T,K)$, where
\begin{equation}
\label{eq:approx-taylor}
     P_{(T,K)}\left(m\right) =  P_{(T,K)}(0)+\frac{P_{(T,K)}^{(2)}(0)}{2!}m^2+\frac{P_{(T,K)}^{(4)}(0)}{4!}m^4+\frac{P_{(T,K)}^{(6)}(0)}{6!}m^6 + \mathcal{O}(m^8), 
\end{equation}
such that the expansion terms are given by
\begin{equation*}
   P_{(T,K)}(0)=\frac{2}{\sqrt T}\Phi^{-1}\Bigg(\sum_{n=1}^{N_q}\lambda_n \Phi\bigg(\frac12\eta_n \sqrt{T}\bigg)\Bigg),
\end{equation*}
\begin{equation*}
    P_{(T,K)}^{(2)}(0)=\frac{1}{2\sqrt{T}}\bigg\{-\frac{1}{\Sigma_0}+\sum_{n=1}^{N_q} \lambda_n \frac{E_n}{H_n}\bigg\},
\end{equation*}
\begin{equation*}
\begin{aligned}
    P_{(T,K)}^{(4)}(0)=&\frac{1}{8\sqrt T}\Bigg\{\frac{1+6\Sigma_2+\Sigma_0^2\big(-7-6\Sigma_2+3\Sigma_2^2\big)}{\Sigma_0^3}+ \sum_{n=1}^{N_q} \lambda_n\bigg[\frac{E_n}{H_n^3}\big(-1+7H_n^2\big)\bigg]\Bigg\},
\end{aligned}
\end{equation*}
\begin{equation*}
\begin{aligned}
    P_{(T,K)}^{(6)}(0)=&\frac{1}{32\sqrt T}\Bigg\{\frac{-3-45\Sigma_2+\Sigma_0^2\big(90\Sigma_2+60\Sigma_4\big)+\Sigma_0^4\Sigma_2\big(45\Sigma_2+60\Sigma_4-15\Sigma_2^2\big)}{\Sigma_0^5}\\
    &\qquad\quad+\frac{16\Sigma_0^2-90\Sigma_2^2-31\Sigma_0^4-45\Sigma_0^2\Sigma_2^2-\Sigma_0^4\big(15\Sigma_2+60\Sigma_4\big)+15\Sigma_0^2\Sigma_2^3}{\Sigma_0^5}\\
    &\qquad\quad+\sum_{n=1}^{N_q}\lambda_n\bigg[\frac{E_n}{H_n^5}\big(3-16H_n^2+31H_n^4\big)\bigg]\Bigg\},
\end{aligned}
\end{equation*}
with the auxiliary quantities:
\begin{equation*}
\begin{aligned}
    &\Sigma_0:=\frac{1}{2}P_{(T,K)}(0)\sqrt{T},\qquad\Sigma_2:=P_{(T,K)}(0)P_{(T,K)}^{(2)}(0)T,\qquad\Sigma_4:=P_{(T,K)}(0)P_{(T,K)}^{(4)}(0)T,\\
    & H_n:=\frac{1}{2}\eta_n\sqrt{T},\qquad E_n:=\exp{\Big(\frac{1}{2}(\Sigma^2_0-H^2_n)\Big)}, \qquad \eta_n = \hat\sigma(T,K; \overline{p}(\theta_n)).
\end{aligned}
\end{equation*}
The expression $\mathcal{O}(m^8)$ is the usual Big-O notation as $m \to 0$ and the subset $A$ is the radius of convergence of the expansion.  This means that for the remainder $R(m)$ we have $\lim_{m\to 0} \frac{R(m)}{m^8} \leq C$ for some constant $C$. 
\end{theorem}
\begin{proof}
Let $(T,K) \in \Pi$ be fixed. We consider the following equation:
\begin{equation}
\label{eq: GeneralEq0}
    f(m,P_{(T,K)}(m))=g(m),
\end{equation}
with $f,g$ given as above, and $P_{(T,K)}(m)$ a continuously differentiable function, such that $P_{(T,K)}(m(T,K)) = \hat\sigma(T,K; \overline p^*, N_q)$. The values of $P_{(T,K)}$ are defined through the functional equation. Since the inverse of $f(\cdot)$ in the second variable is unknown, it is not possible to explicitly express $P_{(T,K)}$ in analytical form, but we apply the technique of \emph{implicit differentiation} to derive a polynomial expansion for it. Differentiating both sides of the \Cref{eq: GeneralEq0}, we have
\begin{equation*}
    f_x(m,P_{(T,K)}(m)) + f_y(m,P_{(T,K)}(m)) P_{(T,K)}'(m) = g'(m),
\end{equation*}
where $f_x,f_y$ are the partial derivatives of $f$ with respect to the first and second input variable. In the domain $D_0:=\{m\in \R|f_y(m,P_{(T,K)}(m))\neq 0\}$ we can write:
\begin{equation*}
    P_{(T,K)}'(m) = \frac{g'(m) - f_x(m,P_{(T,K)}(m))}{f_y(m,P_{(T,K)}(m)) }, \qquad m \in D_0.
\end{equation*}
We can now obtain any order derivative of $P_{(T,K)}(m)$ by differentiating both sides by $m$ and express the $n$-order derivatives of $P_{(T,K)}(m)$ as partial derivatives of $f$ of at most $n$ degrees. Note, however, that the number of derivatives on the right-hand side grows exponentially\footnote{More precisely, the growth is exponentially bounded from above (by $2^{i-1}$ with $ i>0$ the order of differentiation). It is possible to reorder and collect the terms in such a way the overall number is slightly lower, but in any case, the growth is ``more than" polynomial.} in the order of the target derivative (due to the \emph{formula of the derivation of a product}). As we show in \Cref{lem:der_vanish}, the function $P_{(T,K)}$ is even, and therefore, all odd-order derivatives vanish, simplifying the calculation significantly. 

Finally, we obtain the expansion by combining the terms and evaluating $P_{(T,K)}(m)$ at $m=0$:
\[f\left(0,P_{(T,K)}(0)\right) = S_0 \bigg[\Phi\bigg(\frac12P_{(T,K)}(0) \sqrt{T}\bigg)-\Phi\bigg(- \frac12P_{(T,K)}(0) \sqrt{T}\bigg)\bigg] ,\]
and
\[    g(0) = S_0 \sum_{n=1}^{N_q}\lambda_n \bigg[\Phi\bigg(\frac12\eta_n \sqrt T \bigg)-\Phi\bigg(-\frac12\eta_n \sqrt T\bigg)\bigg],\]
from which we obtain 
\begin{equation*}
   P_{(T,K)}(0)=\frac{2}{\sqrt T}\Phi^{-1}\Bigg(\sum_{n=1}^{N_q}\lambda_n \Phi\bigg(\frac12\eta_n \sqrt{T}\bigg)\Bigg).
\end{equation*}
Combining the results, we can express the function $P_{(T,K)}(m)$ as its Taylor expansion function around $0$, which yields (\ref{eq:approx-taylor}), and obtain the value for $\hat\sigma(T,K; \overline p^*, N_q)$ by evaluating at $m(T,K)$. 
\end{proof}
\begin{figure}[H]
    \centering
    \includegraphics[width=0.7\linewidth]{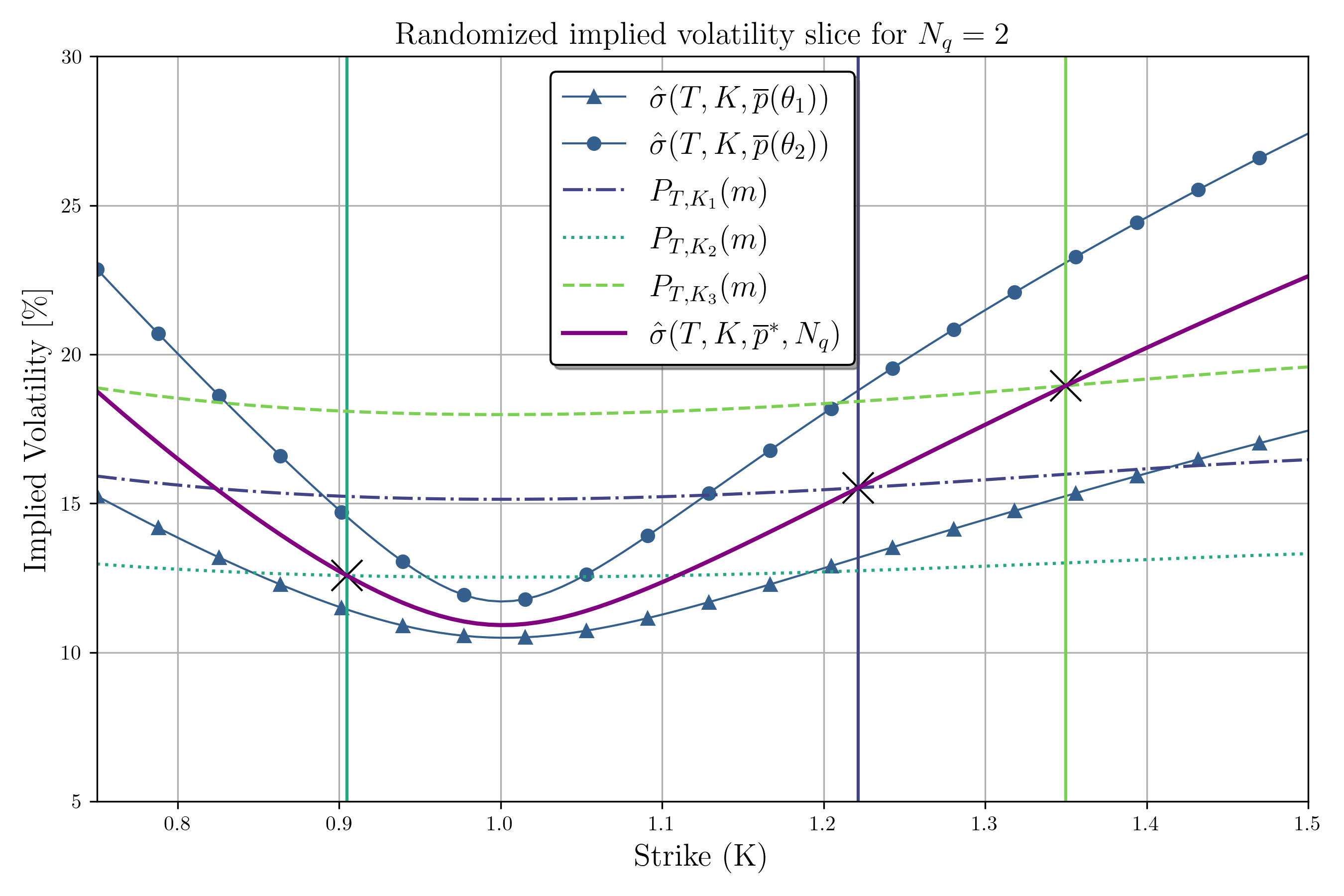}
    \caption{The figure shows the randomization of a parametrization with $N_q=2$ for a fixed time $T$. The grey plots show the individual slices that are mixed, and the green lines show the function $P_{(T,K)}$ for three different $K$s. For each $P_{(T,K)}$, the randomized volatility meets the expansion $P_{(T,K_i)}$ exactly at $K_i$.}
    \label{fig:expansion}
\end{figure}
The expansion formula completes the process of obtaining the randomized surface from a regular parametrization of a volatility surface. A graphical representation of the expansions is shown in \Cref{fig:expansion}. We also summarize the steps of the randomization in \Cref{fig:flow}, which provides an overview of the entire process from the initial parametrization to the randomized parametrization.
\begin{figure}[H]
    \centering
    \includegraphics[width=0.9\linewidth]{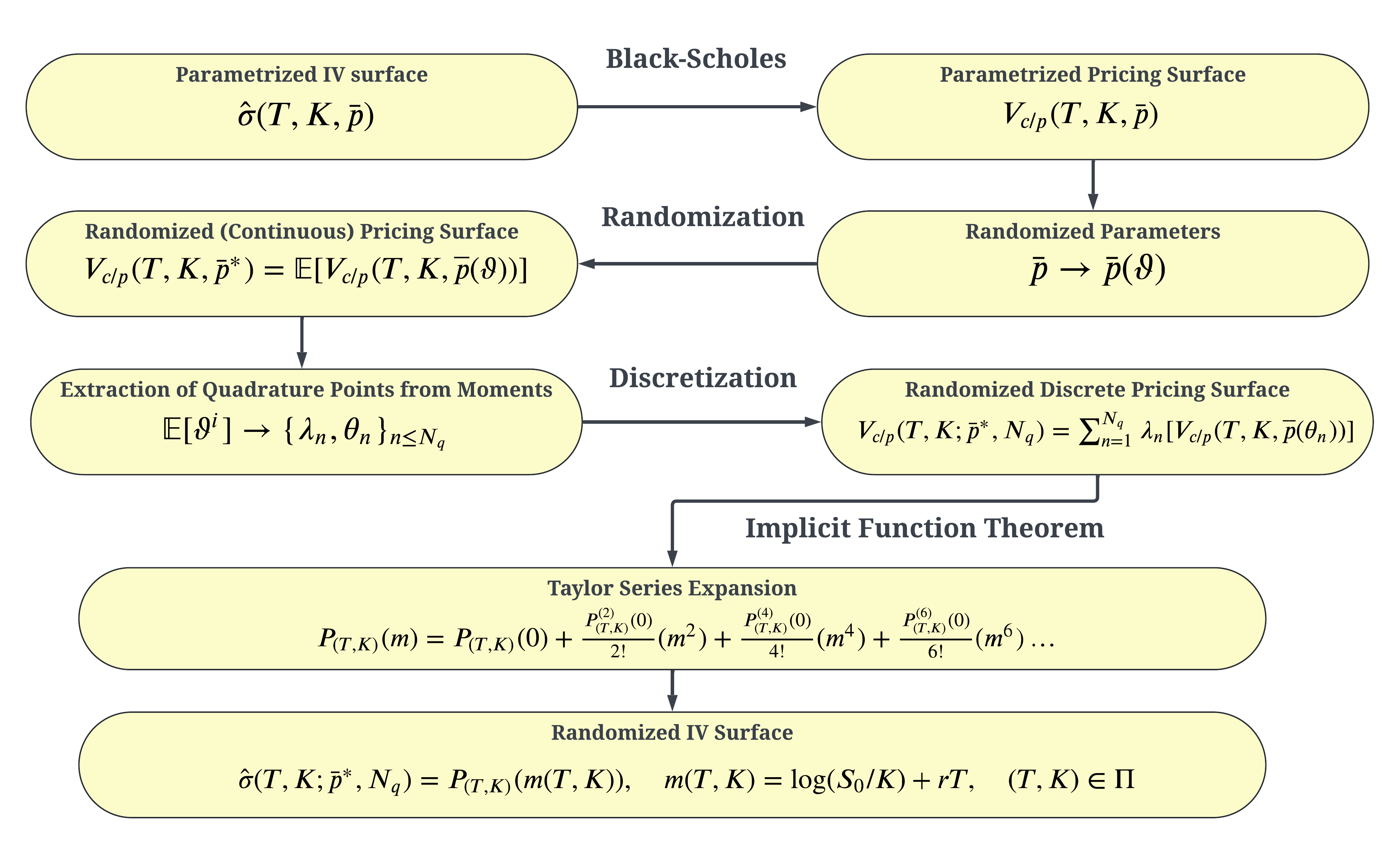}
    \caption{Process of randomization in steps}
    \label{fig:flow}
\end{figure}

\section{Illustrative Examples}
\label{sec:numerics}
The process of randomization transforms a parametric implied volatility surface into a new parametric implied volatility surface with extra parameters and flexibility. We examine a few examples of this randomization process to better understand its mechanics and effectiveness. A repository with a Python implementation is available on \href{https://github.com/NFZaugg/ImpliedVolatilityRandomization}{GitHub}.
\subsection{Randomized Flat Volatility}
 We start by examining the most trivial example of a parametrization, which we call the \emph{flat parametrization}. This parametrization is given as the constant function $\hat\sigma(T,K; \sigma) = \sigma$, where $\sigma \geq 0 $. Since the parametrization is independent of time $T$ and strike $K$, the resulting implied volatility surface is a flat plane with a level of $\sigma$, and the parameter space $\mathcal{D}$ is given by $\mathcal{D} = [0,\infty)$. This parametrization is equivalent to pricing options under the Black-Scholes model with volatility $\sigma$, which is known to not fit well to usual market conditions. We aim to increase its flexibility by randomizing the parameter $\sigma$ and substituting it with a suitable random variable $\vartheta$. Since $\sigma$ must be positive, the random variable $\vartheta$ must be chosen almost surely positive. We propose that $\vartheta$ follows a log-normal distribution with mean $\mu$ and variance $\nu^2$, such that  $\log (\vartheta) \sim \mathcal{N}(\mu,\nu^2)$ for some parameters $\mu, \nu$ in the new parameter space $\mathcal{D}^* = \R \times [0,\infty)$. The randomized pricing surface is then given by \Cref{eqn:genericPricing}, a mixture of Black-Scholes prices weighted by the probability density function of the lognormal distribution
\begin{equation}
    V_{c/p}(T,K; (\mu, \nu)) = \E\left[V_{c/p}(T,K; \vartheta )\right] = \int_0^\infty BS_{c/p}(t_0,S_0,T,K;\sigma) f_{L\mathcal{N}(\mu,\nu^2)} (\sigma) \d \sigma,
\end{equation}
where
\begin{equation}
     f_{L\mathcal{N}(\mu,\nu^2)} (x)  = \frac{1}{x \nu\sqrt{2\pi}} \exp{\left( - \frac{(\log x - \mu)^2}{2\nu^2}\right)}.
\end{equation}
We obtain the discretized randomization price $V_{c/p}(T,K; (\mu, \nu), N_q)$ for $N_q$ quadrature points by substituting the integral in the equation with a finite sum\footnote{Note that this a similar result as Brigo and Mercurio ~\cite{brigo2002lognormal}, since the randomized price is a mixture of Black-Scholes prices.}
\begin{equation}
    V_{c/p}(T,K; (\mu, \nu), N_q) = \sum_{n=1}^{N_q}\lambda_n BS_{c/p}(t_0,S_0,T,K;\sigma_n).
\end{equation}
The quadrature points $\sigma_n$ and quadrature weights $\lambda_n$ can be obtained by computing the moments $\mu_i = {\E[\vartheta^{i}]=\e^{i\mu +i^{2}\nu ^{2}/2}}$ for all $i \leq 2N$ and computing the matrices $M, R$ and $J$ as described in \ref{sec:appendix_qp}. To obtain the implied volatility surface of the randomized parametrization, we can use a root-finding algorithm to find the implied volatility for each option price given by $V_{c/p}(T,K; (\mu, \nu), N_q)$, or we can derive the expansion terms from \Cref{eq:approx-taylor} to obtain an analytical expression for the implied volatility surface $\hat{\sigma}(T,K; (\mu,\nu))$. We examine a slice of the randomized surface for two different sets of parameters. \Cref{fig:randomized-flat} show the results for the expansions with different numbers of coefficients and the ``exact" implied volatility obtained through a root-solving algorithm. In the experiment, we use $r=2\%, T=3$ and $N_q=4$ quadrature points.
\begin{figure}[H]
    \centering
    \includegraphics[width=\linewidth]{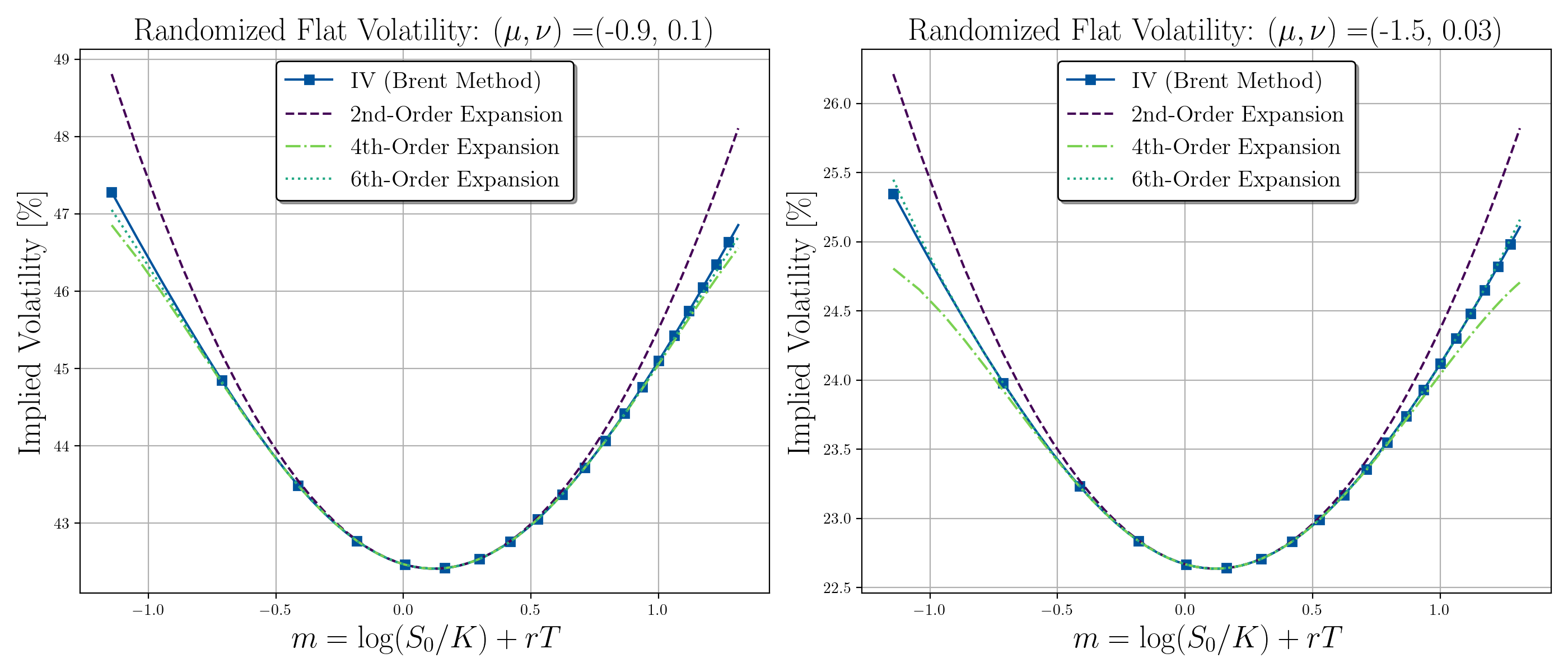}
    \caption{Randomized volatility slice for two different sets of parameters.}
    \label{fig:randomized-flat}
\end{figure}
The randomization of the previously flat volatility surface yielded a new parametrization with two parameters $\mu,\nu$, which is able to form the characteristic ``volatility smile". The result is a variant of the lognormal mixture model~\cite{brigo2002lognormal}, which is known to be able to form a volatility smile. The main difference between the approaches is that the randomization defines the volatility smile using only two parameters, $\mu$ and $\eta$ (excluding the choice of $N_q=4$), offering a more parsimonious alternative to the lognormal mixture model, which uses 7 parameters\footnote{Four volatilities and four weights, where one of the weights is given by the condition that they sum to 1.}. While this may limit some flexibility compared to the discrete mixture model, it simplifies calibration and reduces computational complexity.

\Cref{fig:randomized-flat} shows that the 6th-order approximation is able to closely match the exact implied volatility obtained from the root-finding algorithm. We observe that the fit is worse on the right-hand side for the lower-order expansions, although the shape of the implied volatility smile is similar. The reason is that the overall volatility level is lower on the right-hand side. This apparent increase in the importance of the higher-order coefficients can be explained by examining \Cref{eq:approx-taylor}, where each term is divided by a power of $\Sigma_0 = \frac{1}{2}P_{(T,K)}(0)\sqrt{T}$. 
\begin{rem}[Radius of Convergence]
The Taylor expansion of implied volatility around $m=0$ is subject to a possible finite radius of convergence, beyond which the series may diverge even for an infinite number of expansion terms. While the exact radius is parameter-dependent and theoretically inaccessible, our numerical experiments in \Cref{fig:randomized-flat} indicate that the approximation converges within the plotted range. For deeper OTM options, the accuracy of the expansion should be compared to the implied volatilities obtained from a root solving algorithm.
\end{rem}
The advantage of using the $n$-order expansion versus the application of a root-finding algorithm to the randomized pricing surface is based on the computation complexity. We examine the computational time for an increasing number of pairs $(T,K)$. The root finding algorithm uses Brent's method and terminates once the relative difference between two estimates is at most $1e-08$. If available, the method uses the implied volatility from the previous (usually neighboring) pair of $(T,K)$ to increase the efficiency. \Cref{tbl:rand_time} shows the comparison of the 2nd, 4th and 6th order approximation to the Brent's method. The analytical method demonstrates a negligible increase in computation time as the number of strike/time pairs increases due to its computational efficiency and independence from iterative procedures required by root-finding algorithms like Brent's method.
\begin{table}[H]
    \centering
    \begin{tabular}{l|l|l|l|l}
    \hline
        Number of $(T,K)$ & $10^3 $& $10^4$ &$5.0\cdot10^4$& $10^5$ \\ \hline\hline
        Brent (s) & 0.102 & 1.1336 &9.889& 26.52 \\ \hline
        2nd-order expansion (s) & 0.0009 & 0.001& 0.001 & 0.001 \\ \hline
        4th-order expansion (s) & 0.0006 & 0.0006& 0.0007 & 0.0007 \\ \hline
        6th-order expansion (s) & 0.0006 & 0.0007 & 0.0007& 0.0008 \\ \hline
    \end{tabular}
    \caption{Comparison of analytic expansion vs Brent's method for increasing amount of $(T,K)$: The analytical method is not affected by the increase in strikes, while the Brent method is expected to be $\mathcal{O}(n)$, although even slower in practice due to the increasing memory usage.}
    \label{tbl:rand_time}
\end{table}
We conclude that the randomization of the flat volatility surface using a lognormal random variable allows us to create volatility smiles by adding a single parameter. Although the shape of the smile can be altered by changing the parameter or even using an alternative distribution for $\vartheta$, the additional flexibility is limited (see also~\cite{brigo2002lognormal} for a discussion on this issue). 
\subsection{Randomized SABR parametrization}
In the second example, we consider the \emph{SABR parametrization}, a well-known volatility parametrization introduced by Hagan et al.~\cite{Hagan:2002}. The parametrization is derived from the SABR model\footnote{The SABR model is a stochastic volatility model with the same parameters, such that \begin{align*}
     \d F_t &= \sigma_t F_t^\beta \d W_t,   \\
     \d \sigma_t &= \gamma \sigma_t \d Z_t,
\end{align*} with $\sigma_0 = \alpha$ and the stochastic drivers such that $\d \langle W_t,Z_t\rangle = \rho \d t$. } and given by the formula
\begin{align}
    \label{eq:hagan}
    \nonumber \hat{\sigma}_{H}(T,K; (\alpha,\beta,\rho,\gamma)) = \frac{\alpha}{ \left( F \cdot K \right)^{\frac{1 - \beta}{2}} 
    \left( 
        1 + \frac{(1 - \beta)^2}{24} \log^2{\left(\frac{F}{K}\right)} 
        + \frac{(1 - \beta)^4}{1920} \log^4{\left(\frac{F}{K}\right)} 
    \right) }\cdot \left(\frac{z}{x(z)}\right)\\
   \cdot \left( 
    1 + \left(
        \frac{(1 - \beta)^2}{24} \cdot \frac{\alpha^2}{(F \cdot K)^{1 - \beta}} 
        + \frac{1}{4} \cdot \frac{\rho \cdot \beta \cdot \gamma \cdot \alpha}{(F \cdot K)^{\frac{1 - \beta}{2}}} 
        + \frac{2 - 3\rho^2}{24} \cdot \gamma^2 
    \right) \cdot (T-t_0) 
\right),
\end{align}
where $F = \e^{r(T-t_0)} S_0$ is the $T$-forward of the underlying and 
\begin{align}
    z &= \frac{\gamma}{\alpha} \cdot \left( F\cdot K \right)^{\frac{1 - \beta}{2}} \cdot \log{\left(\frac{F}{K}\right)}, \quad\quad  x(z) = \log{\left( \frac{\sqrt{1 - 2\rho z + z^2} + z - \rho}{1 - \rho} \right)}.
\end{align}
The SABR parametrization is defined by the set of 4 parameters $p=(\beta,\alpha,\rho,\gamma)$ on the parameter space 
\begin{equation}
    (\beta,\alpha, \rho, \gamma) \in \mathcal{D} =[0,1] \times [0,\infty) \times (-1,1) \times [0,\infty).
\end{equation}
Although the parametrization defines an entire volatility surface, market practice is to use the SABR parametrization ``slice-wise", meaning that the calibration is done per volatility slice $\{K \mapsto \hat\sigma_{T_n}( K; \overline{p}), n \leq N\}$ for the set of $N$ expiries observed in the market. The surface is then constructed by a linear interpolation, which is not guaranteed to be free of arbitrage. The details on the interpolation and the conditions are provided later.

Since the SABR parametrization is derived from a stochastic volatility model, the parametrization struggles to fit certain market scenarios in which the market does not follow the parametric regime, such as short-term index option chains. Almost perfect calibration is often impossible for these instances as the parameters reach their limits. We will use the methodology of parameter randomization on the SABR parametrization to increase its flexibility and show that with the help of randomization, the new parametrization will be able to fit the market better. In particular, the randomization substitutes the constant parameter $\gamma$ with a Gamma random variable set by two parameters $k,\theta$. The remaining parameters $\beta,\alpha, \rho$ are not randomized but remain deterministic, yielding a randomized parametrization of 5 parameters $\overline{p}^* = (\beta,\alpha,\rho,k,\theta)$. The parameters $k >0,\theta >0 $ are the shape and scale parameter of a Gamma random variable $\vartheta \sim \Gamma(k,\theta)$ with probability density function
\begin{equation}
    f_{\vartheta(k,\theta)}(x) = \frac{1}{\Gamma(k) \theta^k}x^{k-1}\e^{-x/\theta},
\end{equation}
where $\Gamma(k)$ is the Gamma function. The random variable is almost surely positive, which makes it suitable for a randomization of $\gamma$, since the domain of $\gamma$ is $[0,\infty)$. In this case, the randomized price function is given by
\begin{equation}
        V_{c/p}(T,K; \left(\beta,\alpha,\rho,k, \theta\right))= \int_0^\infty V_{c/p}(T,K;\left(\beta,\alpha,\rho,\gamma\right)) f_{\vartheta(k,\theta)}(\gamma)\d \gamma.
\end{equation}
We transform the continuous randomization into the discretized model. The moments of the Gamma distribution are given by $\E[\vartheta^i] = \theta^i \frac{\Gamma(k+i)}{\Gamma(k)}$ for any $0 \leq i \leq 2N_q$ and we obtain the randomized option price function as
\begin{equation}
\label{eq:rand_hagan_prices}
        V_{c/p}(T,K; \left(\beta,\alpha,\rho,k, \theta\right), N_q)= \sum_{n=1}^{N_q}\lambda_n V_{c/p}(T,K;\left(\beta,\alpha,\rho,\gamma_n\right)),
\end{equation}
where $\{\lambda_n, \gamma_n\}_{n=1}^{N_q}$ are the quadrature weights and points of the Gamma distribution with parameters $k,\theta$. The implied volatility $\hat{\sigma}(T,K,\left(\beta,\alpha,\rho,k, \theta\right), N_q)$ of the randomized SABR parametrization can be obtained by solving the inverse problem to obtain an exact solution, or by using the expansion of \Cref{thm:expansion}. \Cref{fig:skew1} shows the effect of the new parameters $k, \theta$ on the implied volatility shapes. 
\begin{figure}[H]
    \centering
    \includegraphics[width=\linewidth]{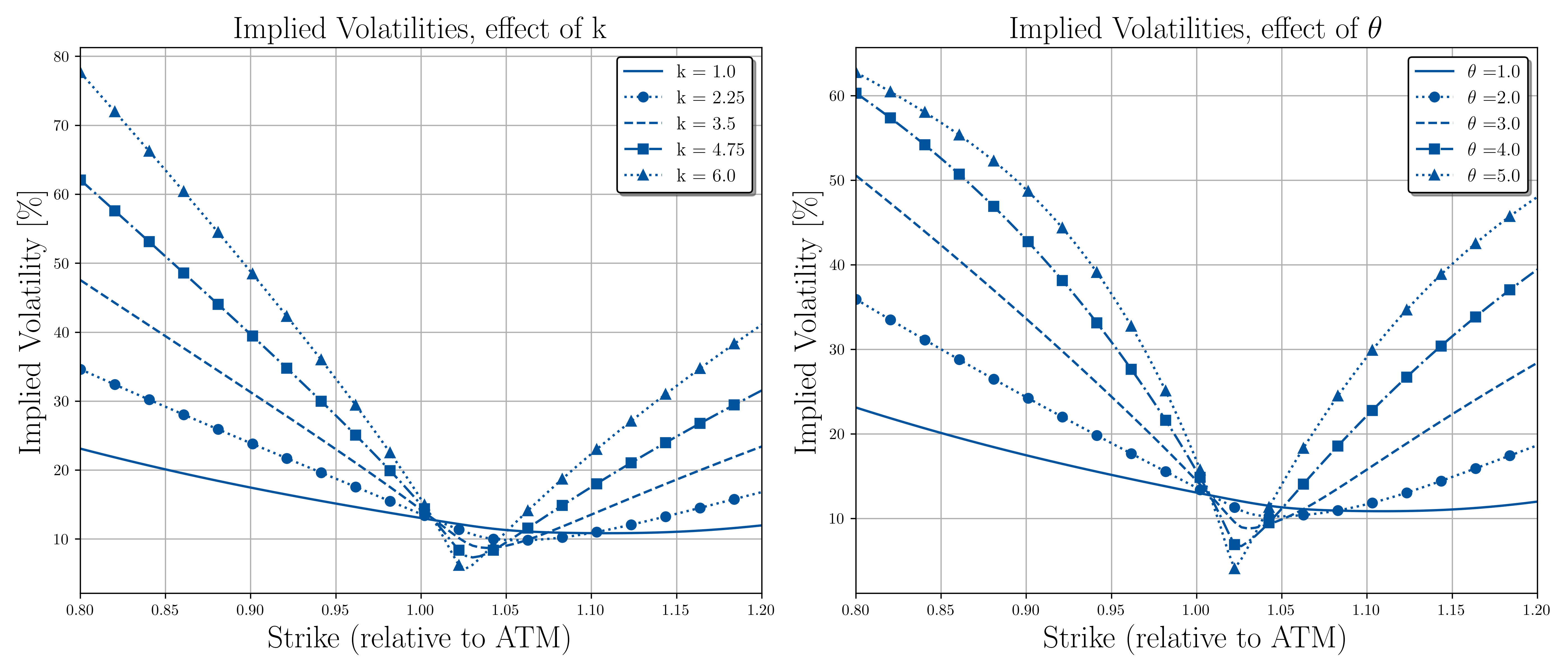}
    \caption{Implied volatility skew: The parameters $k$ and $\theta$ have a strong effect on the shape of the implied volatility. Left: $\left(\beta,\rho,\theta\right) = (0.9,-0.7,1)$, Right: $\left(\beta,\rho,k\right) = (0.9,-0.7,1)$. The parameter $\alpha$ is chosen such that the ATM-vol equals $0.2$. For both plots, we chose $N_q=10$.}
    \label{fig:skew1}
\end{figure}
Although \Cref{fig:skew1} demonstrates how the new parameters influence the skew of the implied volatility curves, it remains unclear how the randomization of the $\gamma$ parameter specifically enhances the fit of the traditional SABR parametrization. In the standard SABR parametrization, the $\gamma$ parameter directly influences the skew of the implied volatility profile. After re-centering to the ATM volatility, a higher $\gamma$ results in a more pronounced skew. It widens the difference between the maximum and minimum implied volatilities within a given strike range. Since $\gamma$ is the sole parameter governing skew, control over the exact curvature of the skew is inherently limited.

By randomizing $\gamma$, we add parameters that modulate not only the level of the skew but also the curvature of the shape around the ATM point. To evaluate the impact of this modification, we conduct the following experiment: using the randomized SABR parametrization, we generate two sets of implied volatility quotes across 40 strikes with parameters $(\beta,\alpha,\rho,k,\theta) = (0.9,0.03,-0.1,0.5,2)$ and $(0.9,0.03,-0.1,0.1,10)$ and $N_q=10$. We then apply an optimizer to fit the traditional SABR parametrization to these quotes, resulting in nearly identical optimal parameter values for both sets, specifically $(\beta, \alpha, \rho, \gamma) = (0.9, 0.25, -0.135, 3.5)$.
\begin{figure}[H]
    \centering
    \includegraphics[width=\linewidth]{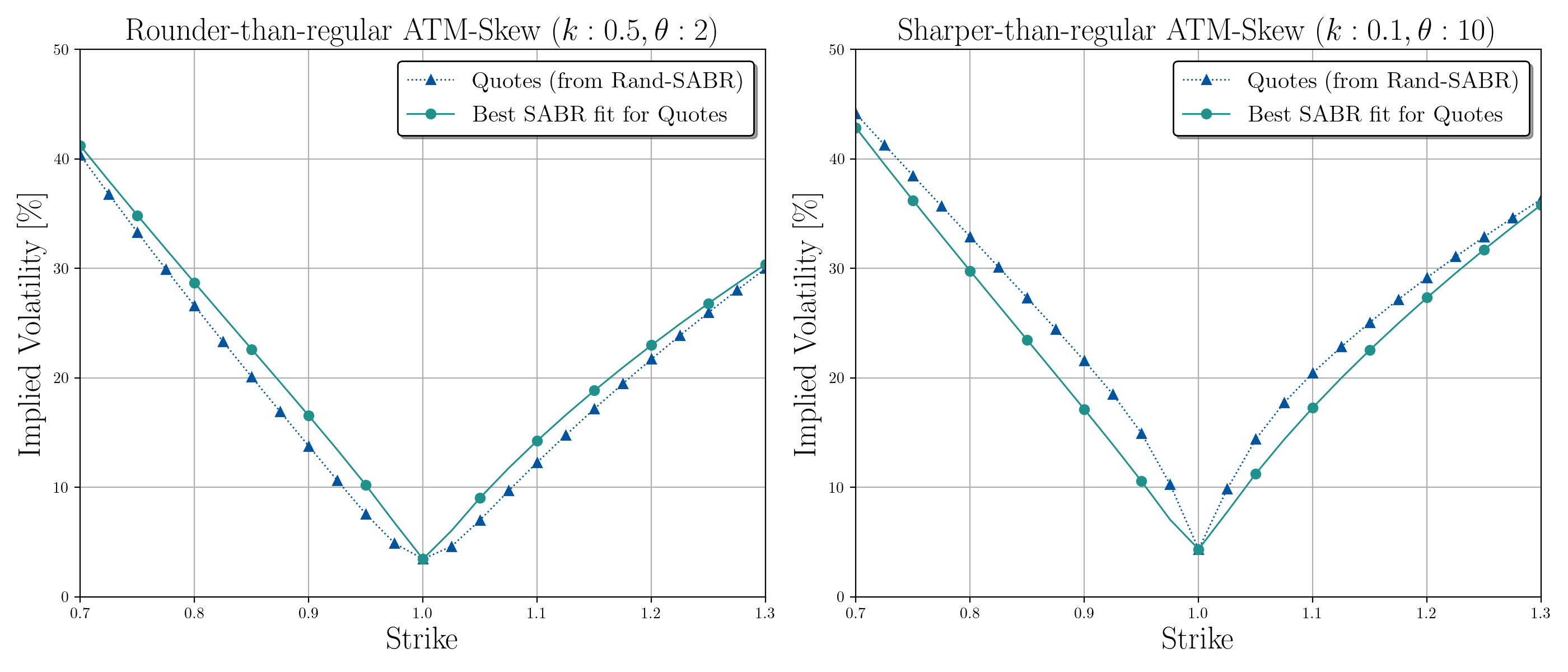}
    \caption{Comparison SABR vs Rand-SABR: the randomized parameters determine the curvature around the ATM point, while leaving the amount of the skew constant.}
    \label{fig:skew-under}
\end{figure}
The experiment in \Cref{fig:skew-under} shows that the SABR parametrization cannot account for the detailed curvature produced by the randomized SABR. The parametrization is limited by the choice of a single parameter. The randomization of the parameter helps to gain flexibility, which can be essential in fitting the market, as we will see in the next subsection. 
\subsection{Fitting the model to the market}
To study the flexibility of the new randomized SABR parametrization, we apply it to real data to examine its fit. We collect a set of option quotes of index options on the SPX index and obtain the market prices for European put and call options on the 31\textsuperscript{st} of July 2024 for a range of strikes and expiries. The data is downloaded from the Cboe data shop~\cite{cboeSPX}, which is based on actual transactions occurring on the Cboe exchange. Information about the data and how it is processed is found in \ref{sec:appendix_results}.

For the randomized SABR parametrization, we choose $N_q=2$ quadrature points and choose a Gamma distribution for $\vartheta$ as described before. We fix $\beta = 0.9$ to reduce the number of parameters to calibrate, following common practice and empirical observations that this value provides a good balance between model flexibility and calibration stability~\cite{Oosterlee_Grzelak_2020}. The remaining parameters are fitted to the market quotes using an optimization scheme. The optimization can be achieved on the randomized prices given by \Cref{eq:rand_hagan_prices} or on the implied volatility surface using the expansion of \Cref{thm:expansion}. \Cref{tbl:Params} shows the calibrated parameters for the randomized SABR parametrization. Since we choose $N_q=2$, the surface is a mixture of two price surfaces with different $\theta_n$ and we report the respective values to show the impact of the Gamma parameter $k,\theta$ on the mixture parameters.  
\begin{rem}[Choice of $N_q$]
Generally, a larger choice of $N_q$ means that the discretized randomization approximates the continuous randomization better, and $N_q=2$ does not seem like a good approximation. However, since the discretization is arbitrage-free for any $N_q \in \N$, an exact approximation of a continuous $\vartheta$ is not necessary, and $N_q$ can be treated simply as a parameter to further specify the nature of the randomization of the parameters. In this case, $N_q=2$ offers an excellent fit. 
\end{rem}
\begin{table}[H]
    \centering
    \begin{tabular}{c || c|c|c|c|c||c|| c|c|c|
    c}
    \textbf{Expiry} & \multicolumn{5}{c||}{ \textbf{Randomization Parameters}}& \textbf{Var }& \multicolumn{4}{c}{\textbf{Mixture Parameters}}\\ \hline
        & ${\beta}$ & $\alpha$ & $\rho$ & $k$ & $\theta$ & Var $\vartheta$ & $\theta_1$ & $\theta_2 $& $\lambda_1$ & $\lambda_2$ \\ \hline\hline
       Aug 24   & 0.9 & 0.322 & -0.595 & 2.379 & 1.04 & 2.572 & 1.602 & 5.424 & 0.772 & 0.228 \\ \hline
                        Sep 24   & 0.9 & 0.308 & -0.617 & 10.587 & 0.187 & 0.37 & 1.531 & 2.804 & 0.647 & 0.353 \\ 
                        Oct 24   & 0.9 & 0.306 & -0.63 & 11.724 & 0.145 & 0.247 & 1.33 & 2.366 & 0.64 & 0.36 \\ 
                        Nov 24   & 0.9 & 0.325 & -0.649 & 9.987 & 0.144 & 0.206 & 1.102 & 2.054 & 0.651 & 0.349 \\ 
                        Dec 24   & 0.9 & 0.329 & -0.657 & 14.387 & 0.09 & 0.116 & 1.031 & 1.737 & 0.627 & 0.373 \\ 
                        Jan 25   & 0.9 & 0.328 & -0.673 & 15.439 & 0.078 & 0.094 & 0.968 & 1.601 & 0.623 & 0.377 \\ 
                        Feb 25   & 0.9 & 0.332 & -0.69 & 15.401 & 0.071 & 0.079 & 0.883 & 1.461 & 0.623 & 0.377 \\ 
                        Mar 25   & 0.9 & 0.335 & -0.689 & 20.794 & 0.051 & 0.053 & 0.865 & 1.337 & 0.607 & 0.393 \\ 
                        Jun 25   & 0.9 & 0.341 & -0.695 & 16.435 & 0.053 & 0.046 & 0.702 & 1.144 & 0.62 & 0.38 \\ 
                        Sep 25   & 0.9 & 0.347 & -0.681 & 15.715 & 0.048 & 0.037 & 0.609 & 1.004 & 0.622 & 0.378 \\ 
                        Dec 25   & 0.9 & 0.35 & -0.656 & 15.016 & 0.045 & 0.031 & 0.543 & 0.904 & 0.625 & 0.375 \\ 
                        Dec 26   & 0.9 & 0.359 & -0.599 & 460.433 & 0.001 & 0.001 & 0.508 & 0.558 & 0.523 & 0.477 \\ 
                        Dec 27   & 0.9 & 0.377 & -0.595 & 681.713 & 0.001 & 0 & 0.398 & 0.43 & 0.519 & 0.481 \\ 
    \end{tabular}
    \caption{Parameter calibration Randomized SABR}
    \label{tbl:Params}
\end{table}
To compare the fit, we run the optimization on the regular SABR parametrization and a lognormal mixture (according to \cite{brigo2002lognormal}) as a benchmark. For the lognormal mixture we use 4 mixture terms and a displacement parameter to obtain the best possible fit. This yields an 8-parameter optimization problem (displacement, 4 volatilities, 3 weights), which we fit using a global optimizer. The fit of the two models is shown in \Cref{fig:skew} compared to the market quotes for the short time-to-maturities and in \Cref{fig:compare-longterm} for the long time-to-maturities. \Cref{tab:mse} shows the observed mean-squared errors of the implied volatilities over the entire dataset.  We see that the randomized SABR parametrization has an excellent fit and can replicate the skew observed in the short-term maturity options on the SPX index. To show the consistency of the improved fit over time we repeat the experiment for the option chains at the two consecutive months ends - 2024-08-30 and 2024-09-30. The additional fitting results are presented in \ref{sec:appendix_results}.
\begin{figure}[H]
    \centering
    \includegraphics[width=\linewidth]{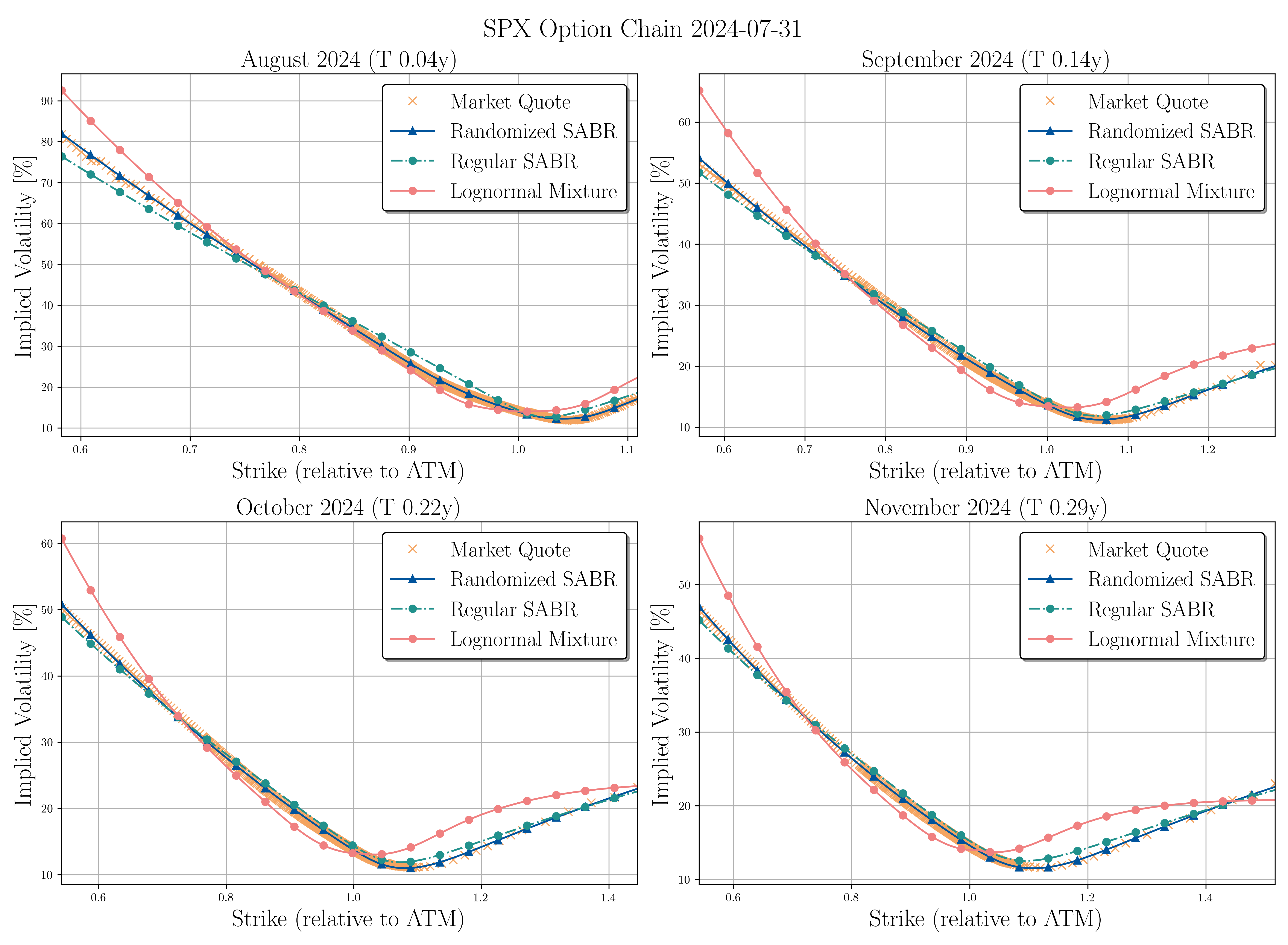}
        \caption{Short time to maturity fit: the graph shows the fit of the randomized SABR vs the benchmarks for short time-to-maturities. }
    \label{fig:skew}
\end{figure}
\begin{figure}[H]
    \centering
    \includegraphics[width=\linewidth]{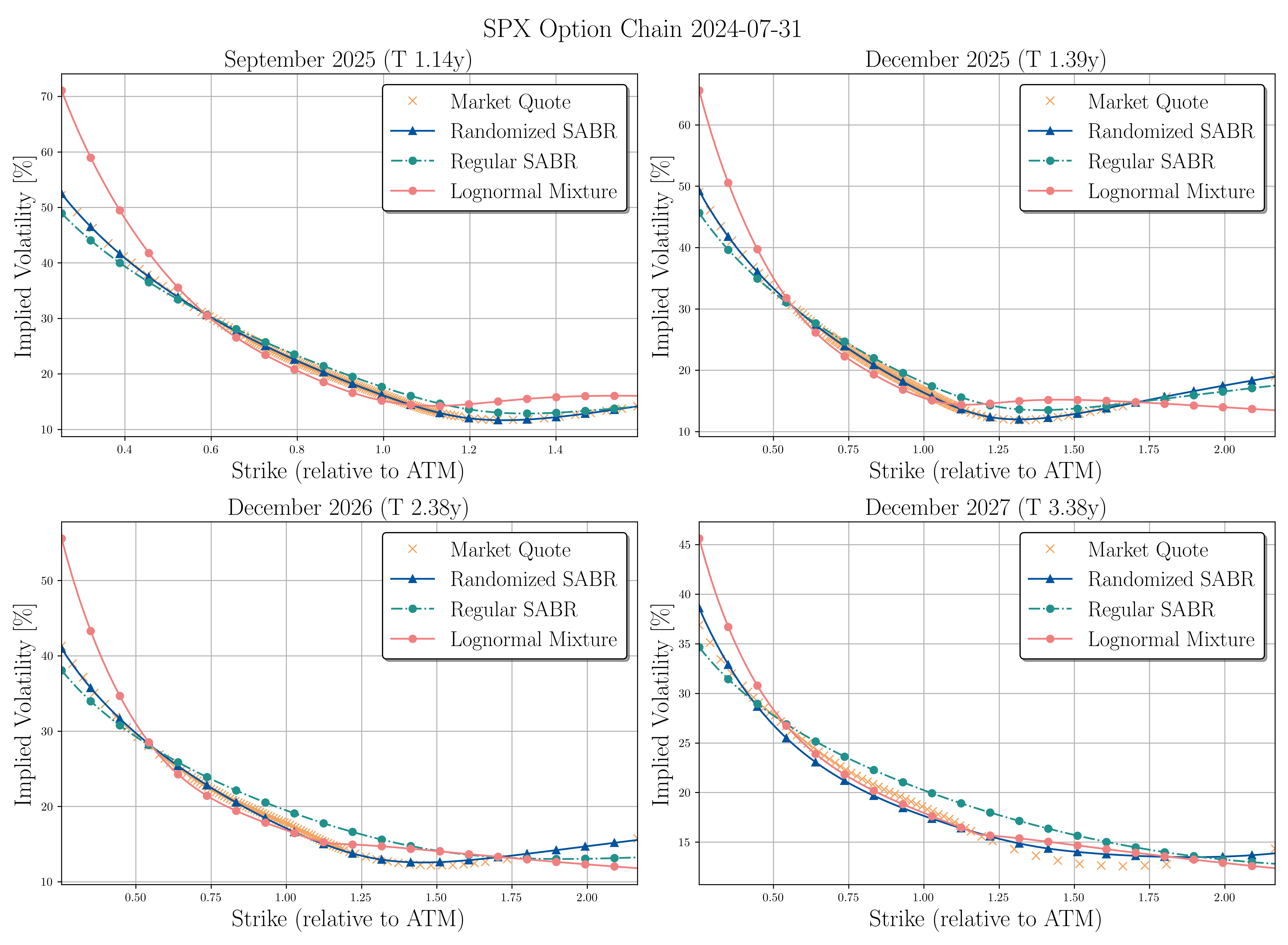}
    \caption{Long time to maturity fit: the graph shows the fit of the randomized SABR vs the benchmarks for long time-to-maturities. }
    \label{fig:compare-longterm}
\end{figure}
\begin{table}[H]
    \centering
    \begin{tabular}{l|l|l|l|}
     \hline
    Jul 31 2024&\multicolumn{3}{c}{\textbf{Mean Squared Errors (MSE)}}\\
      \hline
    \hline
        Expiry Month & Lognormal Mixture& Regular SABR & Randomized SABR  \\ \hline \hline
        Aug 24 & 7.24E-04 & 4.82E-04 & 1.23E-05 \\
        Sep 24 & 8.26E-04 & 1.00E-04 & 6.23E-05 \\
        Oct 24 & 6.73E-04 & 6.74E-05 & 4.69E-05 \\
        Nov 24 & 6.33E-04 & 6.88E-05 & 2.80E-05 \\
        Dec 24 & 6.58E-04 & 6.69E-05 & 1.23E-05 \\
        Jan 25 & 7.61E-04 & 8.02E-05 & 1.01E-05 \\
        Feb 25 & 1.08E-03 & 8.90E-05 & 9.10E-06 \\
        Mar 25 & 1.54E-03 & 8.76E-05 & 7.10E-06 \\
        Jun 25 & 1.65E-03 & 1.31E-04 & 3.20E-06 \\
        Sep 25 & 9.14E-04 & 1.61E-04 & 1.00E-06 \\
        Dec 25 & 7.00E-04 & 1.73E-04 & 1.10E-06 \\
        Dec 26 & 7.07E-04 & 3.32E-04 & 1.20E-05 \\
        Dec 27 & 3.68E-04 & 2.74E-04 & 3.29E-05 \\
         \hline
    \end{tabular}
    \caption{Mean Square Error across entire dataset}
    \label{tab:mse}
\end{table}
We observe from \Cref{tbl:Params} that the variance of the randomized parameter $\vartheta$, which is given by $\Var(\vartheta) = k\theta^2$ decreases as the time to maturity increases. This shows that the longer the time-to-maturity, the less randomization is required to fit this observed market regime. While $\theta_1$ and $\theta_2$ are far apart for the short-expiry months, the values are much closer for longer time-to-maturities. However, from the graphs it is also noticeable that the regular SABR in unable to fit the price surface well.

Since we fit the SABR parametrization slice-wise, the construction of an entire surface requires combining the slices into a unified surface. This is a common practice for the SABR parametrization, and is commonly solved by using a linear interpolation in the \emph{total implied variance}~\cite{bender2020arbitrage}. The interpolation guarantees the absence of calendar arbitrage, but not absence of butterfly arbitrage within the interpolated slices. To illustrate this, suppose that we calibrate a set of slices $\{\hat\sigma_{T_i}(K; \overline{p}_i^*, N_q) : i \leq N\}$, which are arbitrage-free in strike direction, and suppose that there is no calendar-spread arbitrage in the market. This means that, if $T_i \leq T_j$, then
\begin{equation}
    \label{eq:cond-cal}
    V(T_i, K; \overline{p}_i^*, N_q) \leq V(T_j, K; \overline{p}_j^*, N_q), \quad \forall K \in \Pi_K.
\end{equation}
It is well-known~\cite{gatheral2014arbitrage} that the absence of calendar-spread arbitrage is equivalent to the condition that the \emph{total implied variance} is increasing, i.e.
\begin{equation}
    \hat\sigma^2_{T_i}(K; \overline{p}_i^*, N_q)T_i \leq \hat\sigma^2_{T_j}(K; \overline{p}_j^*, N_q)T_j, \quad \forall K \in \Pi_K.
\end{equation}
Under this condition, a linear interpolation in total implied variance suffices to guarantee the absence of calendar arbitrage. Let $a\in [0,1]$ be such that $T = (1-a)T_i + a T_j$ and define 
\begin{equation}
    \hat\sigma(T, K; \overline{p}^*, N_q) = \frac{1}{\sqrt{T}}\sqrt{ (1-a) \hat\sigma^2_{T_i}(K; \overline{p}_i^*, N_q)T_i + a \hat\sigma^2_{T_j}(K; \overline{p}_j^*, N_q)T_j}, \quad \forall K \in \Pi_K.
\end{equation}
Since the interpolation in time does \emph{not} guarantee the absence of butterfly arbitrage in the interpolated slices, each interpolated slice must be checked separately for arbitrage in the strike direction. 

\section{Randomized Spot Volatility Parametrizations \& Near Expiry Options for Equities}
\label{sec:rand-Spot}
In this section of this paper, we consider an extension of the randomization of parametric volatility surfaces to a randomization of the ``spot" parameter of the Black-Scholes formula, rather than implied volatility surface parameters. This new parametrization follows the same principles as the regular randomization described above, and is also defined as an expectation of a randomized price surface. The spot randomization replaces the input of the spot parameter in the Black-Scholes formula with a random variable $\vartheta$, which is centered at the current spot price $S_0$. We show below that, under certain conditions, such a randomization of the spot parameter is arbitrage-free, and that this parametrization is particularly effective to model option markets of very short-term options, known as \emph{near-expiry options}. Near expiry options, also known as \emph{zero-day-expiry options} (0DTE) options\footnote{To be precise, 0DTE options have a time-to-maturity of less than a day. Since we also consider options of 1-3 days expiry, we prefer the expression \emph{near expiry options}.} are options whose maturity date is imminent, i.e., only a few days in the future, or even expire at the end of the day. On the day of an earning announcement, these options chains sometimes exhibit volatility slices with one or more concave sections, such as the W-shaped implied volatility slice or the ``mustache"- shape~\cite{alexiou2023pricing, glasserman2023w}. The unusual shape of implied volatility is an indication of a \emph{bimodal} or \emph{multimodal} risk-neutral probability distribution\footnote{Although note that this is neither a sufficient nor a necessary condition.}, i.e., a probability distribution function that exhibits two modes. The rationale is that the uncertainty from the earnings yields a multimodal risk-neutral probability distribution of $S_T$, reflecting the evolution of the stock price given different scenarios from the earnings (i.e., positive surprise vs negative surprise). Traditional implied volatility parametrizations, which have their origins in stochastic diffusion models, such as the SABR parametrization or the SVI, struggle to produce such shapes since the stochastic driver, which is a diffusion process, is inherently single-modal. The novel randomization method on the spot parameter offers a valuable solution in this case, as the randomization of the spot parameter offers an effective way to create superpositions of the PDFs with varying modes. This procedure yields a mixture density of a similar kind to well-established lognormal mixture model with varying modes~\cite{rebonato2004unconstrained,bloch2011smiling,glasserman2023w}, which were developed as an extension of the classical lognormal mixture (as introduced by Brigo et al. \cite{brigo2002lognormal}), where means of the mixed densities are not the same. We refer to these types of lognormal mixture parametrizations as \emph{multimodal lognormal mixture} (MM LNM) parametrizations. 

In this section we prove that the spot randomized volatility shape is arbitrage-free and show that it produces excellent results on the empirical data, utilizing data from the ticker AMZN on the day of the earnings announcement for Q1 in April 2018. The results show that the spot randomization has a better fit than the MM LNM parametrization, and a significantly better fit than the classical diffusion-based parametrizations SABR and SVI for this particular type of problem. Since the options we consider are short-expiry options, we focus the fitting of the data on a single volatility slice for a fixed time $T$. Although the randomization still provides an entire implied volatility surface, the main purpose of the spot randomization is to fit the non-standard volatility shape and to provide a viable price for any strike for the given expiry.

\subsection{Randomized Spot Parameter}
Let $S_0$ be the current spot price of a financial asset $S$ at time $t_0$. Suppose that $\hat\sigma(T,K,\overline{p})$ is a parametrization with parameters $\overline{p}$ and let $p_{S_T;\overline{p}}$ be the probability density function under a suitable risk-neutral measure $\Q$, given by the parametrization for $S$ at the fixed time $T$. Since the risk-neutral probability density is a PDF of a random variable $S_T$ on a risk-neutral probability space, a necessary condition for it to be arbitrage-free is that $S_T$ is centered at its forward $S_0 \e^{r(T-t_0)}$. This can be shown as if
\begin{equation}
\label{eq:rnd}
    p_{S_T;\overline{p}}(x) = \e^{r(T-t_0)} \frac{\d^2 V_c(t_0,S_0,T,x; \overline{p})}{\d x^2},
\end{equation}
then, an integration of \Cref{eq:rnd} multiplied with $x$ shows that the expectation of $S_T$ is $\E^\Q[S_T] = S_0\e^{r(T-t_0)}$. If we assume $\vartheta$ to be a random variable with law $\mu$ and mean $\E[\vartheta] = S_0$ and assume it is almost surely positive, we can consider a randomization of the spot entry $s$ of the Black-Scholes formula $BS_{c/p}(t_0,s,T,K;\hat\sigma(T,K))$, by replacing $s$ with the random variable $\vartheta$. As it will be presented later, the fact that it is centered at $S_0$ is necessary to ensure that the randomization is arbitrage-free. We find equivalently to \Cref{eq:randPrice}:
\begin{equation}
\label{eq:randPriceSpot}
\E[V_{c/p}(\vartheta,T,K; \overline{p})]=\int_{\R} BS_{c/p}(t_0,\theta,T,K;\hat\sigma(T,K;\overline{p})) \d \mu(\theta),
\end{equation}
where the value of the call/put options is now denoted by $V_{c/p}(S_0,T,K; \overline{p})$ with the extra parameter $S_0$. The randomized price surface \Cref{eq:randPriceSpot} is an average of Black-Scholes prices, where the varying entry is the spot price entry, as opposed to the parameters $\overline{p}$ as before. Ensuring that the spot randomization is given in terms of a set of parameters, we define an extended parameter vector $\overline{p}^* = (p_1,p_2,\dots,p_m,q_1,q_2,\dots) \in \mathcal{D}^*$, which contains all previous parameters plus the parameters $(q_1,q_2\dots)$ specifying the distribution of $
\vartheta$. This defines the pricing surface of the spot randomization, and we can define the randomized volatility surface $\hat\sigma_s( T,K;\overline{p}^*)$ with the new parameter set as the Black-Scholes inverse of \Cref{eq:randPriceSpot}. 
\begin{lem}[Arbitrage-free spot randomization]
    Let $\vartheta$ be centered at $S_0$ and almost surely positive. Then, the randomized spot volatility surface $\hat\sigma_s(T,K;\overline{p}^*)$, given by \Cref{eq:randPriceSpot}, is arbitrage-free.
\end{lem}
\begin{proof}
    We will show that we can write $\E[V_{c/p}(\vartheta,T,K; \overline{p})]$ as an integral of a proper probability density function centered at $S_0\e^{r(T-t_0)}$, which is sufficient to ensure the absence of arbitrage~\cite{brigo2002lognormal}. Suppose that $p_{S_T; \theta }$ is the PDF according to \Cref{eq:rnd}, where $S_0$ is replaced with $\theta$. In this case, we obtain
    \begin{align}
 \nonumber\E[V_{c/p}(\vartheta,T,K; \overline{p})]&=\int_{\R} BS_{c/p}(t_0,\theta,T,K;\hat\sigma_s(T,K;\overline{p})) \d \mu(\theta)\\
 \nonumber&=\int_{\R} \int_K^\infty (x-K) p_{S_T;\theta }(x) \d x \d \mu(\theta)\\
 \nonumber&=\int_K^\infty  (x-K) \int_{\R}  p_{S_T;\theta }(x) \d \mu(\theta) \d x\\
& =\int_K^\infty  (x-K) \E\left[p_{S_T;\vartheta }(x)\right]  \d x.
\end{align}
The expression $x \mapsto \E\left[p_{S_T;\vartheta }(x)\right] $ is a probability density function since it is a convex combination of proper PDFs. It remains to show that it is centered in $S_0 \e^{r(T-t_0)}$, which is required for the absence of arbitrage.
\begin{align}
   \nonumber \int_{0}^\infty x\E\left[p_{S_T;\vartheta }(x)\right]\dx &=  \int_{0}^\infty x \int_{\R}p_{S_T;\theta }(x) \d \mu(\theta) \dx\\
     \nonumber&=   \int_{\R}\int_{0}^\infty x p_{S_T;\theta }(x) \d x \d \mu(\theta)\\
    \nonumber &=   \int_{\R} \theta \e^{r(T-t_0)}\d \mu(\theta) \\
        \nonumber &=    \E[\vartheta]\e^{r(T-t_0)} \\
    &=  S_0\e^{r(T-t_0)} .
\end{align}
This concludes the proof.
\end{proof}
This shows that as long as $\vartheta$ is centered at $S_0$, the randomization is free of arbitrage, as the forward price, whose PDF is given by $x \mapsto \E\left[p_{S_T;\vartheta }(x)\right]$, is centered at $S_0\e^{r(T-t_0)}$. Note that the spot randomization does not model any uncertainty of the spot price $S_0$, which remains a fixed, deterministic input to the model. The randomization is merely defined as randomizing the spot parameter $s$ of the Black-Scholes formula $BS_{c/p}(t_0,s,T,K;\hat\sigma(T,K))$. Intuitively, the randomization introduces uncertainty about potential future trends in the underlying asset's price evolution. Since the Black-Scholes formula implies a diffusion model centered at the forward price (which depends on $r$ and the spot parameter), the risk-neutral PDF of the Black-Scholes model assumes the underlying price at expiry to be centered on the forward price, with symmetric uncertainty around this expectation. This implies the model assigns the highest probability to minimal price movements, such that the mode and mean of the risk-neutral probability density coincide. Such models struggle to accommodate scenarios where large price movements are anticipated, but the direction (up or down) remains uncertain. By randomizing the spot parameter input to the Black-Scholes formula, we create a mixture of multiple Black-Scholes dynamics: the resulting risk-neutral density can become multimodal, with peaks representing distinct market views on possible future mean prices. Crucially, this randomization does not affect the $S_0$, but instead reflects uncertainty about where the price process will center itself in the future. The framework preserves arbitrage-free pricing, as the expected value of the future spot price still matches the forward price.

We now continue the randomization process by discretizing the integral using the quadrature method. We define the discrete randomized pricing surface of the spot parameter as
\begin{equation}
\label{eq:pricing_spot}
     V_{c/p}(T,K; \overline p^*, N_q) :=\sum_{n=1}^{N_q}\lambda_n V_{c/p}(\theta_n,T,K;\overline p),
\end{equation}
where $\{\lambda_n, \theta_n\}_{n \leq N_q}$ are the $N_q \in \N$ quadrature pairs of $\vartheta$.

An interesting example of the spot randomization is the following:
\begin{example}
     Suppose we obtain the risk-neutral probability density function $f_{S_T}(\cdot)$ of the asset price at a fixed expiry $T$. Consider a spot randomization of the parametric surface $\hat\sigma(T,K,\sigma) = \sigma$ with $\sigma=0$ and random spot parameter given by the PDF of $f_{S_T}(\cdot)$. If $\hat\sigma(T,K,\sigma) = 0$, the call option price $V_c(S_0,T,K,\sigma)$ is given by $(S_0- K)^+$, which means that
    \[\E[V_{c}(\vartheta,T,K; 0)] = \int_{\R}V_c(\theta,T,K,\sigma)f_{S_T}(\theta) \d \theta = \int_{\R}(\theta - K)^+f_{S_T}(\theta) \d \theta .\]
    The right-hand side is exactly the call option price function given the risk-neutral probability density function $f_{S_T}(\cdot)$. The randomization perfectly prices any option with expiry $T$.
\end{example}
 The discretized spot randomization is given by (\ref{eq:pricing_spot}). With the same technique of \Cref{thm:expansion} we aim to derive the expansion of the implied volatility surface by setting up an implicit equation 
\begin{equation}
\label{eq:implicit}
    f_\text{spot}(m,P_{(T,K)}(m)) = g_\text{spot}(m),
\end{equation}
for \Cref{eq:pricing_spot}. While the definition of $f_\text{spot}(m, P_{(T,K)}(m))= f(m, P_{(T,K)}(m))$ remains the same as in of \Cref{eq:left-side}, we redefine the function $g(m)$ to incorporate the randomization of the spot instead of the volatility. Note that in the summation of the Black-Scholes formulae, the terms no longer depend on $S_0$ but on $\theta_n$. Writing \begin{equation}
    \log(\theta_n/K) + rT = \log(S_0/K) + \log(\theta_n/S_0) + rT = m+\log(\theta_n/S_0),
\end{equation}
we can define the new right-hand side of \Cref{eq:pricing_spot} to obtain
\begin{equation}
\label{eq:newg}
    g_\text{spot}(m) := S_0 \sum_{n=1}^{N_q} \lambda_n \left[ \frac{\theta_n}{S_0} \Phi\bigg(\frac{m+ \log( \frac{\theta_n}{S_0})+\frac12\eta^2 T}{\eta\sqrt{T}}\bigg)-\e^{-m}\Phi\bigg(\frac{m+\log( \frac{\theta_n}{S_0})-\frac12\eta^2 T}{\eta\sqrt{T}}\bigg)\right],
\end{equation}
where $\eta = \hat\sigma(T,K;\overline{p})$ is the implied volatility from the parametrization. The implicit function theorem then yields:
\begin{theorem}
\label{thm:expansion_spot}
The randomized spot implied volatility surface $\hat\sigma_s(T,K;\overline{p}^*)$ at $(T,K)\in \Pi$ is given by $P_{(T,K)}(m(T,K))$, for $m(T,K) = \log(S_0/K) + rT$, where the function $P_{(T,K)}$ has a Taylor expansion series 
\begin{equation}
    P_{(T,K)}(m) = P_{(T,K)}(0) +  P_{(T,K)}'(0)m + \frac{ P_{(T,K)}^{(2)}(0)}{2!}m^2  + \frac{ P_{(T,K)}^{(3)}(0)}{3!}m^3 + \frac{P_{(T,K)}^{(4)}(0)}{4!}m^4+ \mathcal{O}(m^5),
\end{equation}
with
\begin{align*}
     P_{(T,K)}(0) &= \frac{2}{\sqrt{T}} \cdot \Phi^{-1}\left[ \frac{1}{2} \left( 1+ \sum_{n=1}^{N_q} \lambda_n \Sigma_n\right)\right],\\
     P_{(T,K)}'(0) &= \left[\sum_{n=1}^{N_q} \lambda_n \left( \frac{\Sigma_n}{\eta \sqrt{T}} + \Phi(d_n^-)\right) - \Phi(d_0^-)\right] \cdot \frac{1}{\sqrt{T} \phi(d_0^-)},\\
     P_{(T,K)}^{(2)}(0) &= \left[\sum_{n=1}^{N_q} \lambda_n \left( \frac{\Sigma'_n + 2\phi(d_n^-)} {\eta \sqrt{T}}  -\Phi(d_n^-)\right) - \Sigma'_0\right] \cdot \frac{1}{\sqrt{T}\phi(d_0^-)},
\end{align*}
and auxiliary variables
\begin{align*}
  d^{\pm}_n &:= \frac{\log( \frac{\theta_n}{S_0}) \pm \frac12\eta^2 T}{\eta\sqrt{T}}, \qquad d^{\pm}_0 := {\pm \frac12 P_{(T,K)}(0) \sqrt{T}}, \qquad \Sigma_n := \frac{\theta_n}{S_0}  \Phi(d^+_n) - \Phi(d^-_n),\\
 \Sigma'_n &:= \frac{\theta_n}{S_0}  \Phi(d^+_n) \left(\frac{\log(\theta_n/S_0)}{\eta^2T} + 1/2\right) - \Phi(d^-_n)\left(\frac{\log(\theta_n/S_0)}{\eta^2T} - 1/2\right), \\
   \Sigma'_0 &:= -\Phi(d_0^-) + \phi(d_0^-) \left(\frac{1}{ P_{(T,K)}(0) \sqrt{T}} + \sqrt{T} P_{(T,K)}'(0) - \frac{1}{4}  P_{(T,K)}(0)( P_{(T,K)}'(0))^2 T^{3/2}\right).
\end{align*}
The terms for $ P_{(T,K)}^{(3)}(0), P_{(T,K)}^{(4)}(0)$ are contained in \ref{sec:appendix_order}.
\begin{proof}
    The proof follows the same reasoning as \Cref{thm:expansion}, with the implicit equation given by \Cref{eq:implicit}.
\end{proof}
\end{theorem}
\begin{rem}[Odd-order terms no longer vanish]
    Note that contrary to \Cref{thm:expansion}, the odd-order derivatives do not vanish. This is caused by the fact that the arguments of $\Phi(\cdot)$ in $g_\text{spot}(m)$ are no longer symmetric due to the addition of $\log(\theta_n/S_0)$. This means that the amount of terms in the expansion grows very quickly.
\end{rem}
We further note that the calculation of sensitivities, in particular the delta and gamma, is still feasible under the spot randomization. Since the surface $\hat\sigma_s(T,K; \overline{p}^*)$ is an analytic function of $S_0$, the derivative $\frac{\partial \hat\sigma_s(T,K; \overline{p}^*)}{\partial S_0}$ can be simply derived analytically. In the case a finite difference approach is used, the shocking of the spot price by a quantity $h$ will shift the random variable $\vartheta$ to $\vartheta + h$, which is then centered at $S_0+h$.
\subsection{Illustrative Example: Spot Randomization of Flat Surface}
We again first consider the simplest randomization, which is the flat volatility surface $\hat\sigma(T,K; \sigma) = \sigma$. Instead of replacing $\sigma$ with a random variable, we randomize the spot parameter using a lognormal distributed random variable $\vartheta$, such that $\log(\vartheta) \sim \mathcal{N}(\log(S_0)-\frac{\nu^2}{2} , \nu^2)$ given the parameter $\nu$. The particular choice for the average ensures that \begin{equation}
    \E[\vartheta] = \exp \left(\log(S_0)-\frac{\nu^2}{2} +\frac{\nu^2}{2} \right) = S_0,
\end{equation}
and therefore, that the randomization is arbitrage-free. We compute the prices, implied volatilities, and risk-neutral probability densities for various values of $\nu$ with $N_q=2$. The other parameters are $S_0=3,T=1, \sigma=0.12, r=0\%$ and the results are shown in \Cref{fig:rand_spot_flat}. The third graph shows that the risk-neutral probability density functions exhibit a bimodal shape, while the implied volatility is strongly concave. Next to the exact implied volatilities, we also plot 2nd to 4th-order approximations for the three shapes and highlight the convergence of the approximation to the exact implied volatility curve, which is derived using a root-finding scheme.
\begin{figure}[H]
    \centering
    \includegraphics[width=1\linewidth]{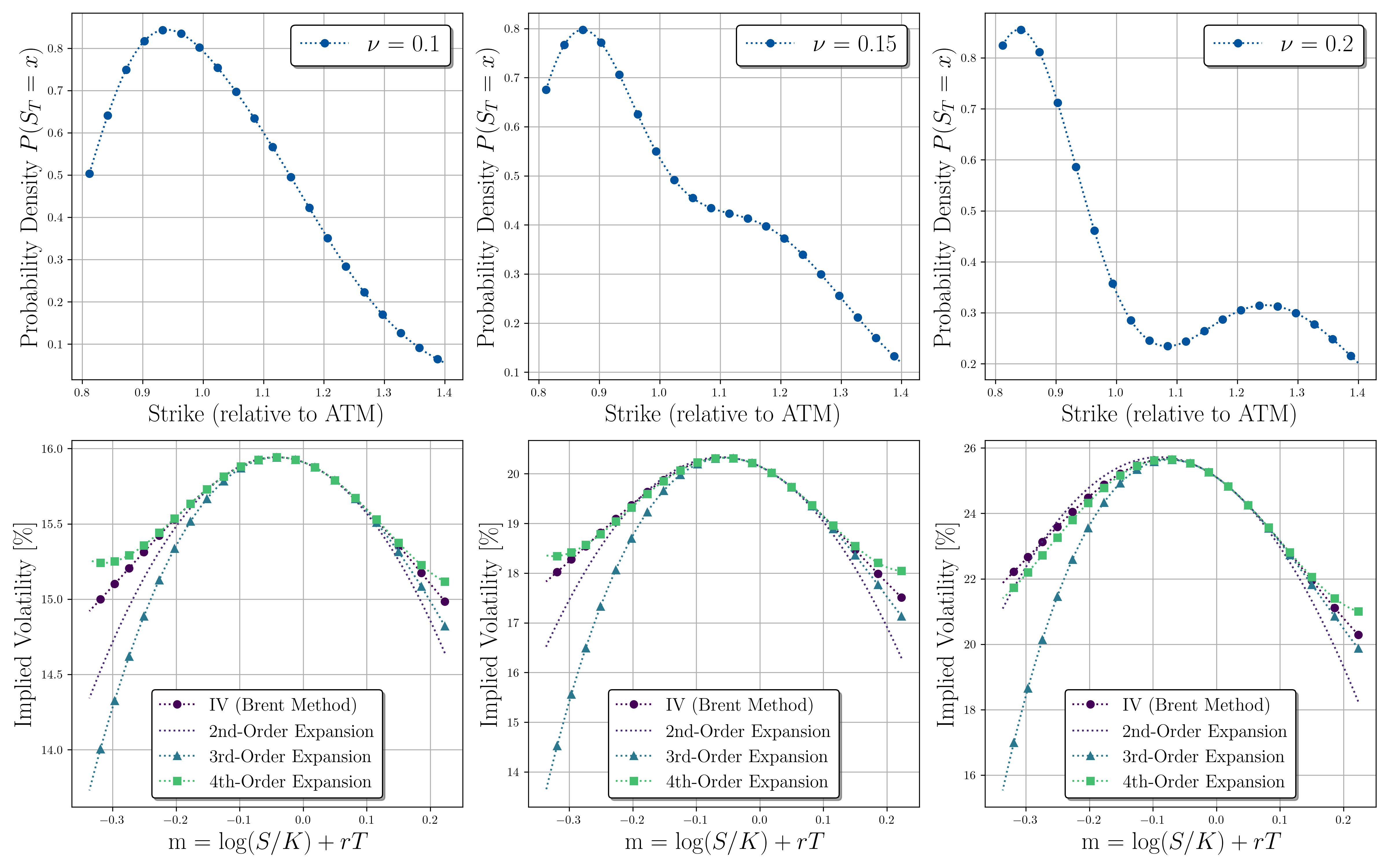}
    \caption{PDF and Implied Volatility of a randomized spot flat volatility surface. The implied volatility is derived from a root finder in the first few orders of approximations. The graph shows the convergence of the expansion to the reference graph. The mixture parameters are $(\lambda_1,\lambda_2,\theta_1,\theta_2) = (0.57,0.43,2.74,3.35), (0.61,0.39,2.64,3.57) $ and $(0.64,0.36,2.55,3.82)$ from left to right.}
    \label{fig:rand_spot_flat}
\end{figure}
We again examine the computational advantage of the analytical method compared to the ``exact" method through a root-finding algorithm (Brent's method). Since the root finder is repeated for every strike, the method is at least $\mathcal{O}(n)$. We run both the 4th order expansion and the root-finding algorithm for increasing numbers of strikes and compare the time used for calculation in \Cref{tbl:rand_spot_time}. 
\begin{table}[H]
    \centering
    \begin{tabular}{l|l|l|l}
    \hline
        Number of strikes & $10^2 $& $10^3$ & $10^4$ \\ \hline\hline
        Brent (s) & 0.23 & 2.36 & 26.61 \\ 
        4th-order expansion (s) & 0.004 & 0.004 & 0.004 \\ \hline
    \end{tabular}
    \caption{Comparison of expansion vs. Brent's method for increasing amount of strikes: The analytical method is not affected by the increase in strikes, while the Brent method is $\mathcal{O}(n)$.}
    \label{tbl:rand_spot_time}
\end{table}
\subsection{Fitting an Earnings Announcement Volatility Surface}
To demonstrate the capabilities of the randomization of the spot parameter, particularly in regard to earnings announcements, we consider an example of real market quotes. The company Amazon.com Inc. announced its earnings release for 2018 Q1 after market close on Apr 26th, 2018, which induced uncertainty in the market on the day of Apr 26th. We obtained the market quotes for options traded on the stock on the morning of the 26th (10:30 AM), which expire the next day on Apr 27th. \Cref{tbl:amzn} shows the summary statistics of the quotes which we obtain from the Cboe Data Shop~\cite{cboeSPX}.
\begin{table}[H]
    \centering
    \begin{tabular}{|l|l|}
    \hline
        Expiry Date & 2018/04/27 \\ \hline
        Spot & 1496.45 \\ \hline
        Min/Max Strike & 1255/1607.5\\ \hline
        N Quotes & 126 \\ \hline
    \end{tabular}
    \caption{Summary Statistics AMZN options}
    \label{tbl:amzn}
\end{table}
Using the quotes, we fit a randomized spot SABR parametrization on the options quotes, where the spot parameter is randomized with a log-normally distributed random variable $\vartheta$, such that $\log \vartheta \sim \mathcal{N}(\log(S_0)-\frac{\nu^2}{2}, \nu^2)$. For the discretization, we choose $N_q = 2$ and run an optimization algorithm to obtain the optimal parameters for $\overline{p}^* = (\beta,\alpha, \rho,\gamma,\nu)$. As a benchmark, we also fit a SABR and SVI-type parametrization, which are not randomized, to the data. Furthermore, we also fit a multimodal lognormal mixture parametrization, which is designed to produce multimodal mixture densities. \Cref{fig:randomized_spot_amzn} shows the fit of the implied volatility parametrizations, the market quotes, and the risk-neutral density from the randomized parametrization. The randomized spot parametrization is able to reproduce the shape of the implied volatility of the quotes. Furthermore, we see from the risk-neutral density this volatility shape stems from a bimodal probability density. The red line in the plot indicates the strike at the mean of the distribution (which is the forward-ATM). The SVI/SABR parametrizations, on the other hand, fail to reproduce the implied volatility shape entirely, fitting to an essentially straight line. 
\begin{figure}[H]
    \centering
    \includegraphics[width=\linewidth]{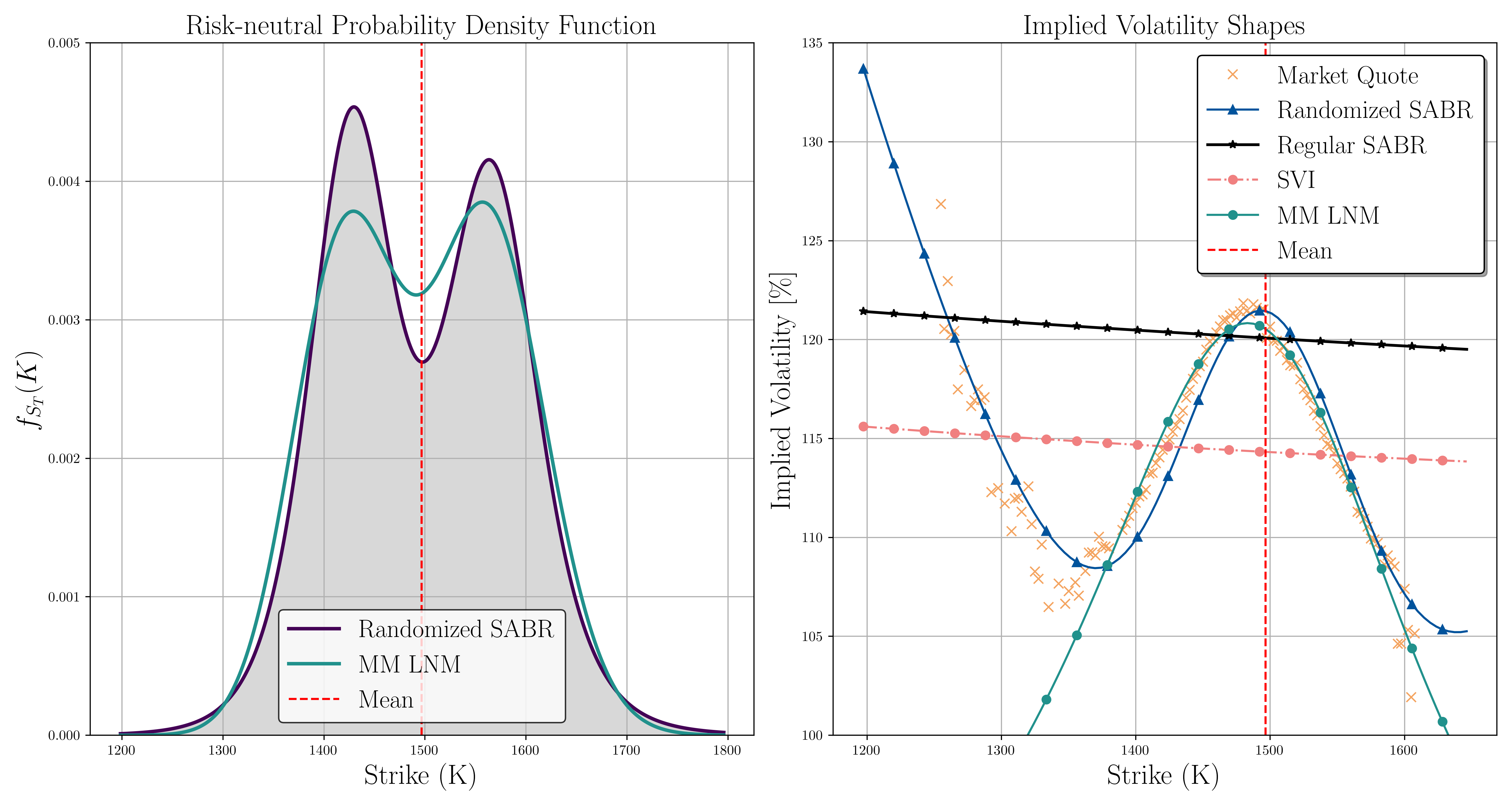}
    \caption{Risk-neutral PDF and implied volatilities of randomized parametrization. The fitted parameters for the randomized SABR are $(\beta,\alpha, \rho,\gamma,\nu) = (0.9,1.38,-0.03,10.99,0.46)$. This is equivalent to $\theta_1, \theta_2 = 1432, 1570$ with weights $\lambda_1,\lambda_2 = 0.53, 0.47$. }
    \label{fig:randomized_spot_amzn}
\end{figure}
\section{Conclusion}
\label{sec:conclusion}
The construction of a clean implied volatility surface is essential for accurately modeling financial derivatives, as it transforms discrete market option prices into a continuous, arbitrage-free representation. A well-defined volatility surface allows traders and risk managers to better understand the pricing dynamics of options across different strikes and maturities, facilitating more informed decision-making in the trading of derivatives. The traditional methods of constructing these surfaces often face challenges, particularly in capturing the nuances of market behavior, especially for options with shorter maturities.

This paper introduces a novel generic method to enhance the flexibility of volatility surface parametrizations through parameter randomization. We formulated new parametric surfaces from existing ones by replacing one of the parameters with a random variable. The method first defines a distribution for a set of parameters and then formulates the randomization as the expectation of the European option price under the distribution. This induces a mixture-type behavior in the pricing surface of the options and through Breeden-Litzenberger~\cite{breeden1978prices} also in the risk-neutral probability density functions. The mixture increases the flexibility of the price and, thus, provides a better ability to fit the market quotes when traditional parametrizations fail. Lastly, we derive an expansion of the implied volatility surface as a function of the input parameters, which can be computed to an arbitrary degree of accuracy. We presented two examples of randomized surfaces and showed that the randomized SABR parametrization is able to fit the data better than a classical SABR parametrization or an SVI-type parametrization. 

In the second part of the paper, we utilized the randomization technique to formulate a randomized spot volatility surface. With this method, the spot parameter of the Black-Scholes formula is randomized using a random variable centered at the current spot price $S_0$. This type of randomization is particularly effective in fitting volatility shapes that imply risk-neutral distribution functions of multi-modal form, which often occur shortly prior to earnings announcements of equities. Again, we derived an expansion of the implied volatility based on the same techniques as before and showed that the parametrization fits well with the data for near-maturity options during earnings announcements, a task that is impossible for classical diffusion-based parametrizations. 
\bibliography{bib.bib}

@article{gatheral2014arbitrage,
  title={Arbitrage-free {SVI} volatility surfaces},
  author={Gatheral, Jim and Jacquier, Antoine},
  journal={Quantitative Finance},
  volume={14},
  number={1},
  pages={59--71},
  year={2014},
  publisher={Taylor \& Francis}
}

@article{Hagan:2002,
  author = {P.S. Hagan and D. Kumar and A.S. Le\'sniewski and D.E Woodward},
  title = {Managing smile risk},
  journal = {Wilmott Magazine},
  year = {2002},
  pages = {84-108}
}

@article{le2021arbitrage,
  title={An arbitrage-free interpolation of class {C2} for option prices},
  author={Le Floch, Fabien},
  journal={Journal of Derivatives},
  volume={28},
  number={4},
  pages={64--86},
  year={2021},
  publisher={Institutional Investor, Inc}
}

@article{mingone2022no,
  title={No arbitrage global parametrization for the {eSSVI} volatility surface},
  author={Mingone, Arianna},
  journal={Quantitative Finance},
  volume={22},
  number={12},
  pages={2205--2217},
  year={2022},
  publisher={Taylor \& Francis}
}

@article{corbetta2019robust,
  title={Robust calibration and arbitrage-free interpolation of {SSVI} slices},
  author={Corbetta, Jacopo and Cohort, Pierre and Laachir, Ismail and Martini, Claude},
  journal={Decisions in Economics and Finance},
  volume={42},
  number={2},
  pages={665--677},
  year={2019},
  publisher={Springer}
}

@article{homescu2011implied,
  title={Implied volatility surface: Construction methodologies and characteristics},
  author={Homescu, Cristian},
  journal={arXiv preprint arXiv:1107.1834},
  year={2011}
}

@article{grzelak2019stochastic,
  title={The stochastic collocation {Monte Carlo} sampler: highly efficient sampling from expensive distributions},
  author={Grzelak, Lech A and Witteveen, Jeroen AS and Suarez-Taboada, Maria and Oosterlee, Cornelis W},
  journal={Quantitative Finance},
  volume={19},
  number={2},
  pages={339--356},
  year={2019},
  publisher={Taylor \& Francis}
}

@article{brigo2002lognormal,
  title={Lognormal-mixture dynamics and calibration to market volatility smiles},
  author={Brigo, Damiano and Mercurio, Fabio},
  journal={International Journal of Theoretical and Applied Finance},
  volume={5},
  number={4},
  pages={427--446},
  year={2002},
  publisher={World Scientific}
}

@article{andreasen2011volatility,
  title={Volatility interpolation},
  author={Andreasen, Jesper and Huge, Brian},
  journal={Risk},
  volume={24},
  number={3},
  pages={76},
  year={2011},
  publisher={Incisive Media Limited}
}

@article{guo2016generalized,
  title={Generalized arbitrage-free {SVI} volatility surfaces},
  author={Guo, Gaoyue and Jacquier, Antoine and Martini, Claude and Neufcourt, Leo},
  journal={SIAM Journal on Financial Mathematics},
  volume={7},
  number={1},
  pages={619--641},
  year={2016},
  publisher={SIAM}
}

@phdthesis{roper2009implied,
  title={Implied volatility: General properties and asymptotics},
  author={Roper, Michael Paul Veran},
  year={2009},
  school={UNSW Sydney}
}

@article{breeden1978prices,
  title={Prices of state-contingent claims implicit in option prices},
  author={Breeden, Douglas T and Litzenberger, Robert H},
  journal={Journal of Business},
  volume={51},
  number={4},
  pages={621--651},
  year={1978},
  publisher={JSTOR}
}

@misc{cboeSPX,
  author = {Cboe},
  title = {{Cboe DataShop}},
  year = 2024,
  howpublished = {https://datashop.cboe.com/},
  note = {Online; accessed 29th September 2024} 
}

@article{alexiou2023pricing,
  title={Pricing event risk: Evidence from concave implied volatility curves},
  author={Alexiou, Lykourgos and Goyal, Amit and Kostakis, Alexandros and Rompolis, Leonidas},
  journal={Swiss Finance Institute Research Paper},
  pages={21-48},
  year={2023}
}

@book{Oosterlee_Grzelak_2020, place={London}, title={Mathematical Modeling and Computation in Finance: With exercises and Python and MATLAB computer codes}, publisher={World Scientific Publishing Europe Ltd}, author={Oosterlee, C. W. and Grzelak, Lech A.}, year={2019}}

@article{glasserman2023w,
  title={W-shaped implied volatility curves and the {Gaussian} mixture model},
  author={Glasserman, Paul and Pirjol, Dan},
  journal={Quantitative Finance},
  volume={23},
  number={4},
  pages={557--577},
  year={2023},
  publisher={Taylor \& Francis}
}

@book{krantz2002implicit,
  title={The implicit function theorem: history, theory, and applications},
  author={Krantz, Steven George and Parks, Harold R},
  year={2002},
  publisher={Springer Science \& Business Media}
}

@article{golub1969calculation,
  title={Calculation of {Gauss} quadrature rules},
  author={Golub, Gene H and Welsch, John H},
  journal={Mathematics of computation},
  volume={23},
  number={106},
  pages={221--230},
  year={1969}
}

@article{fengler2009arbitrage,
  title={Arbitrage-free smoothing of the implied volatility surface},
  author={Fengler, Matthias R},
  journal={Quantitative Finance},
  volume={9},
  number={4},
  pages={417--428},
  year={2009},
  publisher={Taylor \& Francis}
}

@book{gatheral2011volatility,
  title={The volatility surface: A Practitioner's Guide},
  author={Gatheral, J},
  year={2011},
  publisher={John Wiley and Sons, Inc}
}

@article{github,
title = {On randomization of affine diffusion processes with application to pricing of options on {VIX} and {S\&P} 500},
journal = {Applied Mathematics and Computation},
volume = {508},
pages = {129598},
year = {2026},
issn = {0096-3003},
doi = {https://doi.org/10.1016/j.amc.2025.129598},
url = {https://www.sciencedirect.com/science/article/pii/S0096300325003248},
author = {Lech A. Grzelak},
}

@article{bender2020arbitrage,
  title={Arbitrage-free interpolation of call option prices},
  author={Bender, Christian and Thiel, Matthias},
  journal={Statistics \& Risk Modeling},
  volume={37},
  number={1-2},
  pages={55--78},
  year={2020},
  publisher={De Gruyter Oldenbourg}
}

@article{franccois2022venturing,
  title={Venturing into uncharted territory: An extensible implied volatility surface model},
  author={Fran{\c{c}}ois, Pascal and Galarneau-Vincent, R{\'e}mi and Gauthier, Genevi{\`e}ve and Godin, Fr{\'e}d{\'e}ric},
  journal={Journal of Futures Markets},
  volume={42},
  number={10},
  pages={1912--1940},
  year={2022},
  publisher={Wiley Online Library}
}

@article{ackerer2020deep,
  title={Deep smoothing of the implied volatility surface},
  author={Ackerer, Damien and Tagasovska, Natasa and Vatter, Thibault},
  journal={Advances in Neural Information Processing Systems},
  volume={33},
  pages={11552--11563},
  year={2020}
}

@article{rebonato2004unconstrained,
  title={Unconstrained fitting of implied volatility surfaces using a mixture of normals},
  author={Rebonato, Riccardo and Cardoso, Maria Teresa},
  journal={Journal of Risk},
  volume={7},
  pages={55--74},
  year={2004}
}

@article{bloch2011smiling,
  title={Smiling at evolution},
  author={Bloch, Daniel Alexandre and Coello, CA Coello and Securities, Mizuho and others},
  journal={Applied Soft Computing},
  volume={11},
  number={8},
  pages={5724--5734},
  year={2011},
  publisher={Elsevier}
}

@article{wilkens2005option,
  title={Option pricing based on mixtures of distributions: Evidence from the Eurex index and interest rate futures options market},
  author={Wilkens, Sascha},
  journal={Derivatives Use, Trading \& Regulation},
  volume={11},
  pages={213--231},
  year={2005},
  publisher={Springer}
}

@article{chataigner2020deep,
  title={Deep local volatility},
  author={Chataigner, Marc and Cr{\'e}pey, St{\'e}phane and Dixon, Matthew},
  journal={Risks},
  volume={8},
  number={3},
  pages={82},
  year={2020},
  publisher={MDPI}
}

@article{gonon2024operator,
  title={Operator Deep Smoothing for Implied Volatility},
  author={Gonon, Lukas and Jacquier, Antoine and Wiedemann, Ruben},
  journal={arXiv preprint arXiv:2406.11520},
  year={2024}
}

@article{jacquier2019randomized,
  title={The randomized Heston model},
  author={Jacquier, Antoine and Shi, Fangwei},
  journal={SIAM Journal on Financial Mathematics},
  volume={10},
  number={1},
  pages={89--129},
  year={2019},
  publisher={SIAM}
}

@article{jacquier2018black,
  title={Black--Scholes in a CEV random environment},
  author={Jacquier, Antoine and Roome, Patrick},
  journal={Mathematics and Financial Economics},
  volume={12},
  pages={445--474},
  year={2018},
  publisher={Springer}
}
\appendix
\section{Quadrature Pairs $\{ \lambda_n, x_n\}_{n \leq N_q}$ for Expectations}
\label{sec:appendix_qp}
Here, we present the derivation of the quadrature points to calculate the expectation of the randomized pricing surface. The section is based on~\cite{grzelak2019stochastic} and Golub and Welsch~\cite{golub1969calculation}, which establish the algorithm for the computation of the quadrature points for random variables based on the moments. 

The Gaussian quadrature points enable an efficient approximation of the expectation $\E[g(X)]$ given an arbitrary function $g(\cdot)$ and a random variable $X$ with law $\mu$. Since the expectation is an integral over the domain $\mathcal{D}$ of $X$, the integral can be written as:
\begin{equation}
    \E[g(X)] = \int_{\mathcal{D}} g(x) \d \mu(x) \approx \sum_{n=1}^{N_q} \lambda_n g(x_n),
\end{equation}
where $\{ \lambda_n, x_n\}_{n \leq N_q}$ points are chosen in an optimal way, which we establish below. The algorithm is general for any type of expectation and requires only an efficient computation of the moments $\mu_i = \E[X^i]$ of $X$. The foundation of the algorithm is to establish a sequence of \emph{polynomials} $p_0(x),p_1(x),\dots$, which are orthonormal with respect to $X$, i.e., such that
\begin{equation}
    \E[p_i(X)p_j(X)] = \begin{cases} 1 \quad &\text{ if } \quad  i=j,\\
     0 \quad &\text{ otherwise. }
    \end{cases}
\end{equation}
In this case, the quadrature points $\theta_n$ of degree $N_q$ are given by the roots $p_{N_q}(\theta_n) = 0$ of $p_{N_q}(x)$, and the quadrature weights are given by a formula. Finding the quadrature pairs is thus reduced to obtaining a sequence of orthonormal polynomials. The polynomial sequence can be constructed from monomials $m_n(x) = m^n$, whose expectations are exactly equal to the moments $\mu_n$ of $X$. Denoting $\mu_{i,j} = \mu_{i+j} = \E[X^iX^j]$, the orthonormal polynomials can be calculated the following way: first, we compute the $N_q+1$-dimensional Gram-matrix $M$ given by
\begin{equation}
    M = \begin{bmatrix}
\mu_{0,0} & \mu_{0,1} & \cdots & \mu_{0,N_q} \\
\mu_{1,0} & \mu_{1,1} & \cdots & \mu_{1,N_q} \\
\vdots    & \vdots    & \ddots & \vdots    \\
\mu_{N_q,0} & \mu_{N_q,1} & \cdots & \mu_{N_q,N_q}
\end{bmatrix},
\end{equation}
containing all moments until $2N_q$. Since the matrix is symmetric and positive semi-definite, we can use the Cholesky-decomposition $M= R^TR$ to obtain the triangular matrix $R$. Next, we calculate the quantities $\alpha_j$ and $\beta_j$, defined as
\begin{equation}
    \alpha_j = \frac{r_{j,j+1}}{r_j,j} - \frac{r_{j-1,j}}{r_{j-1,j-1}},\quad j = 1,\dots,N_q,\quad \text{and} \quad \beta_j = \left(\frac{r_{j+1,j+1}}{r_{j,j}}\right)^2, \quad j = 1, \dots N_q-1,\end{equation}
where $r_{i,j}$ are the elements of the matrix $R$. Using the coefficients $\alpha, \beta$, we can recursively define the polynomial sequence $p_n(x)$ as 
\begin{equation}
    p_{j+1} = (x-\alpha_j)p_j(x) - \beta_jp_{j-1}, \quad  j = 1, \dots N_q-1,
\end{equation}
with $p_0 \equiv 0, p_1 \equiv 1$. One can show that the polynomials are indeed orthonormal and thus that the quadrature points are given by the roots of $p_{N_q}$. A bit of linear algebra shows that the roots can be found by the eigenvalues of the matrix
\begin{equation}
    J:= \begin{bmatrix} \alpha_1& \sqrt{\beta_1}& 0 &0&\dots&0\\
    \sqrt{\beta_1} &\alpha_2& \sqrt{\beta_2}& 0&\dots&0\\
    &\dots&&&\dots&\\
    0&\dots&0&\sqrt{\beta_{N_q-2}} &\alpha_{N_q-1}& \sqrt{\beta_{N_q-1}}\\
    0&\dots&0&0&\sqrt{\beta_{N_q-1}} &\alpha_{N_q}
    \end{bmatrix},
\end{equation}
and additionally the quadrature weights $\lambda_n = (v^n_1)^2$ by the square of the first row of the $n$-th eigenvalue $v_n$ (see \cite{golub1969calculation} for an extensive discussion). An implementation of the algorithm in Python and MATLAB can be found on \href{https://github.com/LechGrzelak/Randomization}{GitHub}~\cite{github}.

\section{Supplementary Results for SPX Fitting}
In this appendix, we provide additional information and results for the SPX fitting experiment. 

\Cref{tbl:data-summary} summarizes the data of the option chains and expiries available in the set. Included in the set were all options with a valid bid quote at the EOD snapshot for the main expiry dates. When both put and call quotes are available we select the most liquid quote per strike based on the open interest, which are the out-the-money quotes. The quotes were not cleaned for arbitrage before fitting.
    \begin{table}[H]
    \centering
    \begin{tabular}{c|c|c|c|c}
    \hline
        \textbf{Expiry Date} & \textbf{Days to Maturity} & \textbf{N Quotes} & \textbf{Min Strike} & \textbf{Max Strike} \\ \hline\hline
        16-Aug-2024 & 16 & 364 & 3225 & 6140 \\ 
        20-Sep-2024 & 51 & 352 & 3150 & 7100 \\ 
        18-Oct-2024 & 79 & 339 & 3000 & 8000 \\ 
        15-Nov-2024 & 107 & 289 & 3000 & 8400 \\ 
        20-Dec-2024 & 142 & 233 & 2400 & 9000 \\ 
        17-Jan-2025 & 170 & 228 & 2000 & 9000 \\ 
        21-Feb-2025 & 205 & 147 & 1800 & 9000 \\ 
        21-Mar-2025 & 233 & 148 & 1400 & 9000 \\ 
        20-Jun-2025 & 324 & 160 & 1400 & 9000 \\ 
        19-Sep-2025 & 415 & 144 & 1400 & 8800 \\ 
        19-Dec-2025 & 506 & 144 & 1400 & 12000 \\ 
        18-Dec-2026 & 870 & 88 & 1400 & 12000 \\ 
        17-Dec-2027 & 1234 & 58 & 1400 & 12000 \\ 
    \end{tabular}
    \caption{Data summary implied volatilities quotes on Jul 31st, 2024. The strike range is chosen based on sufficiently liquid quotes available.}
    \label{tbl:data-summary}
\end{table}
\Cref{fig:aug24} and \Cref{fig:sep24} show the fitting results for the randomized SABR vs the two benchmarks, and \Cref{tab:mseaug}, \Cref{tab:msesep} show the mean squared errors. In both cases, the randomized SABR performs unanimously the best, matching the market quotes closely.
\label{sec:appendix_results}
\begin{figure}[H]
    \centering
    \includegraphics[width=\linewidth]{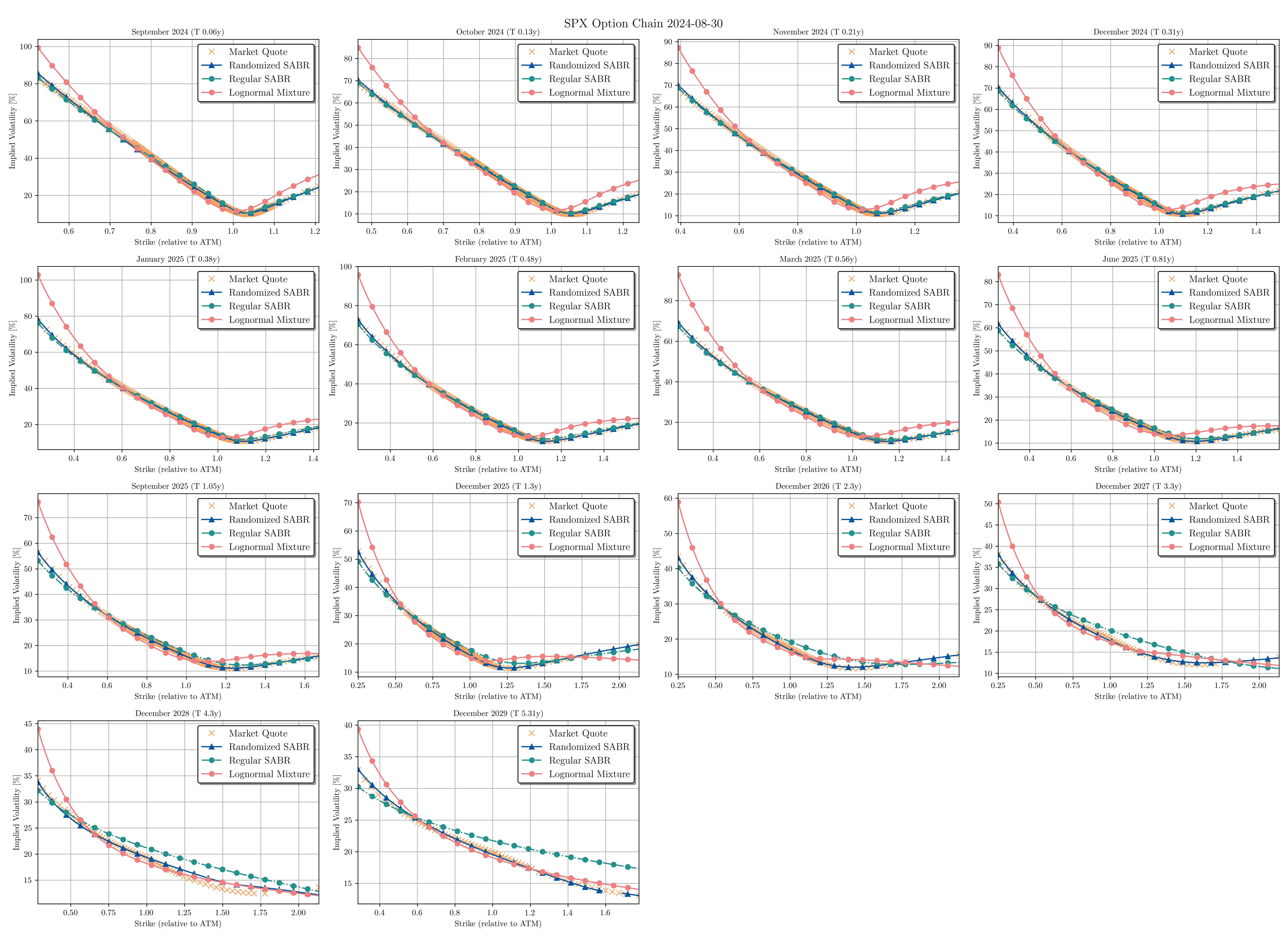}
    \caption{Option Chain for Aug 30 2024}
    \label{fig:aug24}
\end{figure}
\begin{figure}[H]
    \centering
    \includegraphics[width=\linewidth]{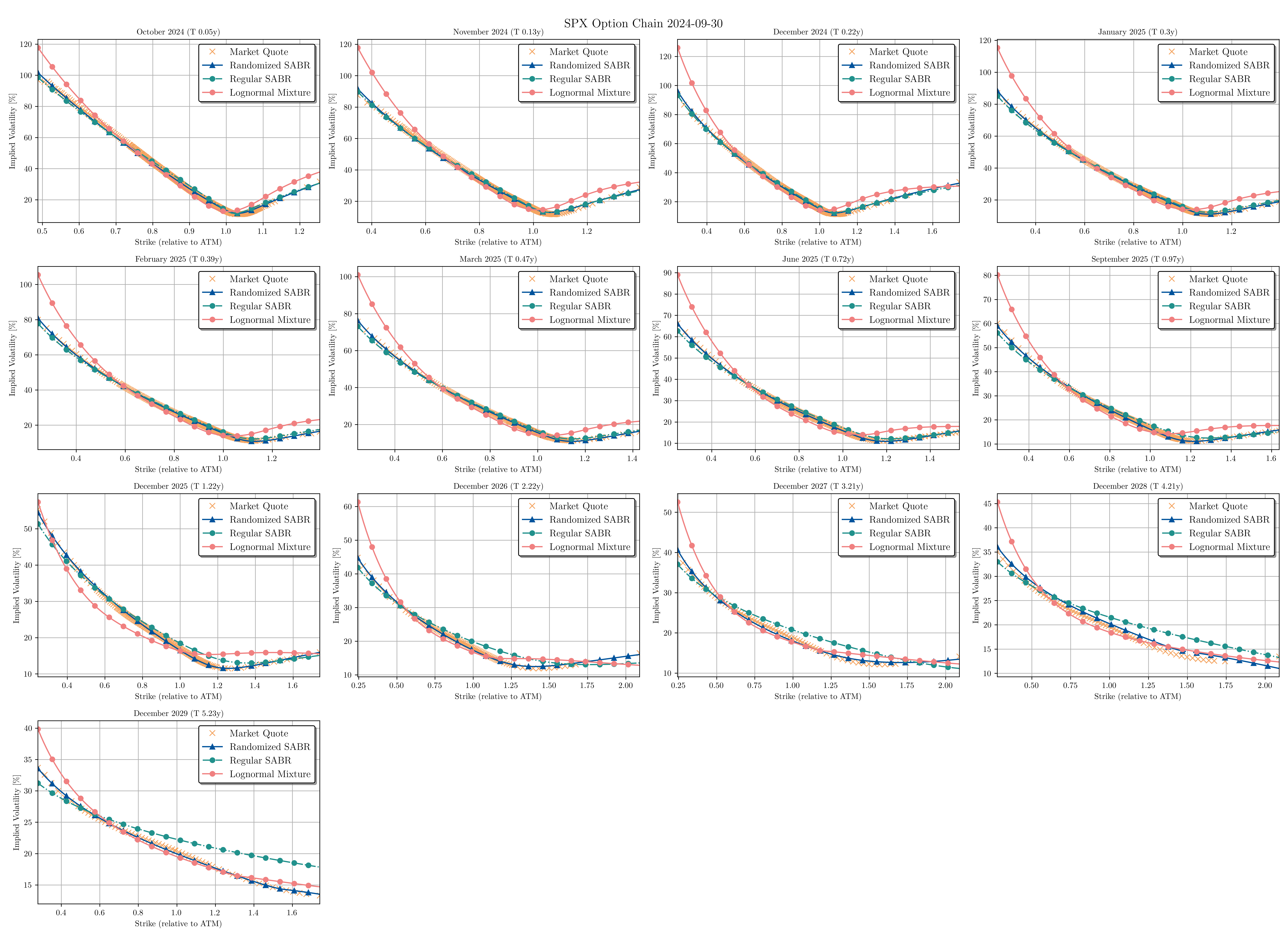}
    \caption{Option Chain for Sep 30 2024}
    \label{fig:sep24}
\end{figure}
\begin{table}[H]
    \centering
    \begin{tabular}{l|l|l|l|}
     \hline
    Aug 30 2024&\multicolumn{3}{c}{\textbf{Mean Squared Errors (MSE)}}\\
      \hline
    \hline
        Expiry Month & Lognormal Mixture& Regular SABR & Randomized SABR  \\ \hline \hline
Sep 24 & 1.14E-03 & 1.68E-04 & 7.59E-05 \\ \hline
Oct 24 & 1.05E-03 & 4.31E-05 & 2.54E-05 \\ \hline
Nov 24 & 1.08E-03 & 4.62E-05 & 1.70E-05 \\ \hline
Dec 24 & 1.05E-03 & 5.62E-05 & 6.90E-06 \\ \hline
Jan 25 & 1.27E-03 & 7.16E-05 & 8.20E-06 \\ \hline
Feb 25 & 1.22E-03 & 8.59E-05 & 5.20E-06 \\ \hline
Mar 25 & 1.81E-03 & 6.89E-05 & 5.10E-06 \\ \hline
Jun 25 & 1.43E-03 & 1.19E-04 & 1.90E-06 \\ \hline
Sep 25 & 1.08E-03 & 1.67E-04 & 2.40E-06 \\ \hline
Dec 25 & 8.33E-04 & 1.91E-04 & 6.00E-06 \\ \hline
Dec 26 & 8.60E-04 & 3.24E-04 & 1.04E-05 \\ \hline
Dec 27 & 7.20E-04 & 3.42E-04 & 1.01E-05 \\ \hline
Dec 28 & 4.62E-04 & 5.01E-04 & 6.07E-05 \\ \hline
Dec 29 & 2.31E-04 & 5.83E-04 & 7.00E-06 \\  \hline
    \end{tabular}

    \caption{Mean Square Error across entire dataset - Aug 30}
    \label{tab:mseaug}
\end{table}
\begin{table}[H]
    \centering 
    \begin{tabular}{l|l|l|l|}
     \hline
    Sep 30 2024&\multicolumn{3}{c}{\textbf{Mean Squared Errors (MSE)}}\\
      \hline
    \hline
        Expiry Month & Lognormal Mixture& Regular SABR & Randomized SABR  \\ \hline \hline
Oct 24 & 1.24E-03 & 2.97E-04 & 7.41E-05 \\ \hline
Nov 24 & 1.78E-03 & 9.29E-05 & 5.65E-05 \\ \hline
Dec 24 & 1.74E-03 & 7.08E-05 & 3.00E-05 \\ \hline
Jan 25 & 1.45E-03 & 9.37E-05 & 8.00E-06 \\ \hline
Feb 25 & 1.40E-03 & 9.36E-05 & 7.30E-06 \\ \hline
Mar 25 & 1.47E-03 & 1.04E-04 & 5.10E-06 \\ \hline
Jun 25 & 1.70E-03 & 1.43E-04 & 5.70E-06 \\ \hline
Sep 25 & 1.18E-03 & 1.83E-04 & 1.84E-05 \\ \hline
Dec 25 & 8.64E-04 & 2.10E-04 & 1.44E-05 \\ \hline
Dec 26 & 1.01E-03 & 3.30E-04 & 1.90E-05 \\ \hline
Dec 27 & 8.44E-04 & 4.01E-04 & 3.37E-05 \\ \hline
Dec 28 & 5.25E-04 & 5.29E-04 & 9.87E-05 \\ \hline
Dec 29 & 2.22E-04 & 6.17E-04 & 6.20E-06 \\ \hline
         \hline
    \end{tabular}
    \caption{Mean Square Error across entire dataset - Sep 30}
    \label{tab:msesep}
\end{table}
\section{Higher-Order (3rd,4th) Terms Calculation for Spot Randomization }
\label{sec:appendix_order}
Here, we present the 3rd and 4th order from the spot randomization of \Cref{thm:expansion_spot}. We derived the terms by separately computing the partial derivatives $f_x = \frac{\partial f_\text{spot}}{\partial x }f_y = \frac{\partial f_\text{spot}}{\partial y }$ of $f_\text{spot}$, its higher order partial derivatives, and the derivative of the function $g_\text{spot}(m)$. Then, the expansion terms are given as
\begin{align*}
    P^{(3)} &=\frac{1}{f_y} \bigg[g^{(3)} - f_{xxx} - 3 f_{xxy} P' - 3 f_{xyy} P'{}^2 - 3 f_{xy}  P^{(2)} - 3 f_{yy} P' P^{(2)} - f_{yyy} P'{}^3\bigg]\\
    P^{(4)} &=\frac{1}{f_y} \bigg[
    g^{(4)} - f_{xxxx} - f_{yyyy} P'{}^4 - 4 f_{xyyy} P'{}^3 - 4 f_{yxxx} P' - 6 f_{xxyy} P'{}^2 - 6 f_{yyy} P'{}^2 P^{(2)} - 12 f_{xyy} P' P^{(2)} \\
    &\quad - 6 f_{xy} P^{(2)} - 4 f_{yy} P^{(3)} P' - 6 f_{yy} P^{(2)} - 4 f_{xy} P^{(3)} \bigg],
\end{align*}
where we have used the short-hand notation 
\begin{equation*}
    P = P_{(T,K)}(0), \quad P' = P_{(T,K)}'(0), \quad P^{(2)} = P_{(T,K)}^{(2)}(0), \quad P^{(3)} = P^{(3)}(0).
\end{equation*}
This yields the equations
\begin{align*}
P^{(3)} &= \sum_{n=1}^{N_q} \Bigg[\frac{1}{16 T^2} \Bigg( \frac{T \left( 48 P' + P \left( 24 + P T \left( 24 P^{(2)} - 6 P P'{}^2 T + P'{}^3 T \left( 4 - P^2 T \right) + 12 P' \left( -1 + P P^{(2)} T \right) \right) \right) \right)}{P^2} \\
&\quad - \frac{24 \e^{\frac{1}{8} T \left( P - \eta \right) \left( P + \eta \right) - \frac{\beta_n^2}{2 T \eta^2}} T \sqrt{\e^{\beta_n}} \lambda_n}{\eta} \\
&\quad - 16 \e^{\frac{P^2 T}{8}} \sqrt{2 \pi} T^{3/2} \left( \Phi\left( \frac{P \sqrt{T}}{2} \right) - \lambda_n \Phi\left( \frac{T \eta^2 - 2 \beta_n}{2 \sqrt{T} \eta} \right) \right) \\
&\quad - \frac{16 \e^{\frac{1}{8} T \left( P - \eta \right) \left( P + \eta \right) - \frac{\beta_n^2}{2 T \eta^2}} \sqrt{\e^{\beta_n}} \lambda_n \beta_n}{\eta^3} \Bigg)\Bigg]\\
P^{(4)} &= \sum_{n=1}^{N_q} \Bigg[\frac{\eta^2}{64 P^3 T^2 \eta^5}   \Bigg( 64 + T \eta^3 \Bigg( -768 P'{}^2 + 384 P \left( -P' + P^{(2)} \right) \\
&\quad + 8 P^5 P'{}^2 \left( P' - 3 P^{(2)} \right) T^3 + P^6 P'{}^4 T^4 - 16 P^2 \left( 7 + 6 P'{}^2 T \right) \\
&\quad + 32 P^3 T \left( P' - 3 P^{(2)} + 4 P^{(3)} - P'{}^2 \left( P' - 3 P^{(2)} \right) T \right) \\
&\quad + 4 P^4 T^2 \left( 24 P^{(2)} + P' \left( 6 P' - 24 P^{(2)} + 16 P^{(3)} - 3 P'{}^3 T \right) \right)  \\
&\quad + 16 \e^{\frac{1}{8} T \left( P - \eta \right) \left( P + \eta \right)} P^3 \left( -4 + 7 T \eta^2 \right) \sqrt{\e^{\beta_n}} \lambda_n \Bigg) \\
&\quad + 32 \Sigma_n P^3 \Bigg( \frac{\e^{\frac{P^2 T}{8}}}{\Sigma_n} \sqrt{2 \pi} T^{5/2} \eta^5 \left( \Phi\left(\frac{P \sqrt{T}}{2}\right) - \lambda_n \Phi\left(\frac{T \eta^2 - 2 \beta_n}{2 \sqrt{T} \eta}\right) \right) \\
&\quad + 4 \Sigma_n \e^{\frac{1}{8} T \left( P - \eta \right) \left( P + \eta \right)} T \eta^2 \sqrt{\e^{\beta_n}} \lambda_n \beta_n \\
&\quad + 2 \e^{\frac{1}{8} T \left( P - \eta \right) \left( P + \eta \right)} \sqrt{\e^{\beta_n}} \lambda_n \beta_n^2 \Bigg)  \Bigg],
\end{align*}
with 
\begin{align*}
    &\beta_n = \log(\theta_n/S_0) \quad \text{ and }\quad  \Sigma_n = \e^{-\frac{\beta_n^2}{2 T \eta^2}}.
\end{align*}
\end{document}